%% file: GORS-IJFS-Arxived-Post.tex
\RequirePackage{fix-cm}

\documentclass[twocolumn]{svjour3}        
\smartqed  
\usepackage{amsmath,amsfonts,amssymb,epsfig,epstopdf,url,array}

\usepackage{amsthm}
\usepackage{color,soul}
\usepackage{algorithmic}
\usepackage{balance}
\usepackage[ruled]{algorithm2e}
\usepackage{stmaryrd}
\usepackage[table]{xcolor}
\usepackage{centernot}
\usepackage{tikz}
\usepackage{booktabs}
\usepackage{pdfpages}
\usepackage{graphicx}
\usepackage{subfigure}
\usepackage[font=small,skip=8pt]{caption}
\usepackage{tabularx}
\usepackage{comment}
\usepackage{multirow}
\usepackage{bigstrut}
\usepackage{pdflscape}
\usepackage{rotating}
\usepackage{float}
\usepackage{lipsum}
\usepackage{enumitem}
\usetikzlibrary{plotmarks,shapes,arrows,chains,hobby,backgrounds,calc,trees}
\usepackage{makecell}
\usepackage{bm}
\usepackage{afterpage}
\usepackage{supertabular}

\usepackage[section=subsection,sort=use]{glossaries}
\usepackage[toc,page]{appendix}
\usepackage[T1]{fontenc}

\newglossarystyle{formel_altlong4colheader}{%
	\setglossarystyle{altlong4colheader}%
}
\newglossarystyle{abbreviationStyle}{%
	\setglossarystyle{long3col}%
	\renewenvironment{theglossary}%
	{\begin{supertabular}[l]{@{}p{0.4\hsize}p{0.9\hsize}p{0.01\hsize}@{}}}%
		{\end{supertabular}}%
}
\newglossarystyle{notationStyle}{%
	\setglossarystyle{long3col}%
	{\begin{supertabular}[l]{@{}p{0.4\hsize}p{0.9\hsize}p{0.01\hsize}@{}}}%
		{\end{supertabular}}%
}

\newglossary{abbreviation}{1i}{1o}{}
\newglossary{notation}{2i}{2o}{}
\loadglsentries{glossary}
\makeglossaries

\newcommand{\head}[1]{\textbf{#1}}
\definecolor{lightgray}{gray}{0.9}
\definecolor{white}{rgb}{1,1,1}
\theoremstyle{definition}
\newtheorem{exmp}{Example}
\theoremstyle{definition}
\newtheorem{mydef}{Proposition}
\theoremstyle{definition}
\newtheorem{defn}{Definition}
\newcommand*{\permcomb}[4][0mu]{{{}^{#3}\mkern#1#2_{#4}}}
\newcommand*{\perm}[1][-3mu]{\permcomb[#1]{P}}

\def\makeheadbox{\relax}

\begin{document}

\title{Modeling and Selection of Interdependent Software Requirements using Fuzzy Graphs}

\author{Davoud Mougouei, David M. W. Powers}


\institute{School of Computer Science, Engineering and Mathematics\\
	Flinders University,
	Adelaide, Australia \\
	Tel.: +61 8 82012841\\
	\email{\{davoud.mougouei,david.powers\}@flinders.edu.au}           
}
\date{This is a post-peer-review, pre-copyedit version of an article published in International Journal of Fuzzy Systems. The final authenticated version is available online at: \url{http://dx.doi.org/10.1007/s40815-017-0364-4}.}

\maketitle

\begin{abstract}
Software requirement selection is to find an optimal set of requirements that gives the highest value for a release of software while keeping the cost within the budget. However, value-related dependencies among software requirements may impact the value of an optimal set. Moreover, value-related dependencies can be of varying strengths. Hence, it is important to consider both the existence and the strengths of value-related dependencies during a requirement selection. The existing selection models however, either assume that software requirements are independent or they ignore strengths of requirement dependencies. This paper presents a cost-value optimization model that considers the impacts of value-related requirement dependencies on the value of selected requirements (optimal set). We have exploited algebraic structure of fuzzy graphs for modeling value-related requirement dependencies and their strengths. Validity and practicality of the work are verified through carrying out several simulations and studying a real world software project.

\end{abstract}
\input{introduction}
\input{related}
\input{modeling}
\input{modeling_whyfuzzygraph}
\input{modeling_interdependency}
\input{factoring}   
\input{factoring_ov}
\input{factoring_gors}
\input{factoring_examples}  
\input{validation} 
\input{simulation}

\input{simulation_design}

\input{simulation_result}
\input{validation_case}

\input{identification}
\input{conclusion}
\bibliographystyle{IEEEtran}
\bibliography{ref}

\vspace{4em}
\begin{appendices}	
    \vspace{1.5em}
	\section{Acronyms}
	The acronyms below are listed based on the order of first appearance in the paper.
	\vspace{-2.5em}
	\printglossary[type=abbreviation,style=abbreviationStyle]
	
	\section{Notations}
	A glossary of the frequently used symbols in this paper is given below.
	\vspace{-2.5em}
	\printglossary[type=notation,style=notationStyle]
\end{appendices}

\end{document}

%% file: glossary.tex
\newglossaryentry{NRP}{type=abbreviation,name=NRP,description=Next Release Problem}
\newglossaryentry{BKP}{type=abbreviation,name=BKP,description=Binary Knapsack Problem}
\newglossaryentry{SDP}{type=abbreviation,name=SDP,description=Selection Deficiency Problem}
\newglossaryentry{EV}{type=abbreviation,name=EV,description=Estimated Value}
\newglossaryentry{CV}{type=abbreviation,name=CV,description=Customer Value}
\newglossaryentry{AV}{type=abbreviation,name=AV,description=Accumulated Value}
\newglossaryentry{OV}{type=abbreviation,name=OV,description=Overall Value}
\newglossaryentry{GORS}{type=abbreviation,name=GORS,description=Graph Oriented Requirement Selection}
\newglossaryentry{FRIG}{type=abbreviation,name=FRIG,description=Fuzzy Requirement Interdependency Graph}
\newglossaryentry{LOI}{type=abbreviation,name=LOI,description=Level Of Interdependency}
\newglossaryentry{RAN}{type=abbreviation,name=RAN,description=Radio Access Network}
\newglossaryentry{PMR}{type=abbreviation,name=PMR,description=Performance Management Traffic Recording}
\newglossaryentry{PMS}{type=abbreviation,name=PMS,description=Precious Messaging System}
\newglossaryentry{R}{type=notation,name=$R$,description=Set of requirements}
\newglossaryentry{r_i}{type=notation,name=$r_i$,description= Requirement $r_i \in R$}
\newglossaryentry{v_i}{type=notation,name=$v_i$,description= Estimated value of $r_i$}
\newglossaryentry{c_i}{type=notation,name=$c_i$,description= Estimated cost of $r_i$}
\newglossaryentry{b}{type=notation,name=$b$,description= Available budget}
\newglossaryentry{x_i}{type=notation,name=$x_i$,description= Selection variable ($r_i \in R$ is selected or not)}
\newglossaryentry{D}{type=notation,name=$D$,description=Set of explicit value-related dependencies}
\newglossaryentry{O}{type=notation,name=$O$,description= Optimal set}
\newglossaryentry{Otilde}{type=notation,name=$\tilde{O}$,description= Set of excluded requirements}
\newglossaryentry{G}{type=notation,name=$ G $,description= A fuzzy requirement interdependency graph}
\newglossaryentry{mu}{type=notation,name=$\mu$,description= Fuzzy membership function of requirements}
\newglossaryentry{rho}{type=notation,name=$\rho$,description= Fuzzy membership function of dependencies}
\newglossaryentry{wedge}{type=notation,name=$\wedge$,description=Fuzzy AND operator}
\newglossaryentry{I_i}{type=notation,name=$I_{i}$,description= Impact of excluded requirements on $r_i$}
\newglossaryentry{vee}{type=notation,name=$\vee$,description=Fuzzy OR operator}
\newglossaryentry{eta}{type=notation,name={$\eta$},description={Pearl's measure of causal strength}}

%% file: introduction.tex
\section{Introduction}
\label{sec_introduction}
Owing to budget constraints, it is hardly if ever feasible to fully implement the requirements of a software release~\cite{bagnall_next_2001}. Therefore, requirement selection is needed to find an optimal set of requirements with the highest value while keeping the cost within the budget~\cite{mougouei2019fuzzy,mougouei2015partial,dahlstedt_moulding_2003,barney_product_2008,ruhe_product_2010,achimugu_systematic_2014,kukreja_value_based_2013,Fogelstr_Usingportfolio_2010}. This problem, also known as the \textit{Next Release Problem (\gls{NRP})}~\cite{bagnall_next_2001}, is mathematically formulated as a \textit{Binary Knapsack Problem (\gls{BKP})}~\cite{mougouei2017integer,harman_exact_2014,szoke_conceptual_2011,zhang_multi_objective_2007,carlshamre_release_2002,lust_multiobjective_2012,karlsson_optimizing_1997,jung_optimizing_1998}. 

Based on the BKP formula, the existing selection models aim to maximize the \textit{Accumulated Value (\gls{AV})} of an optimal set on the assumption that the value of an optimal is derived by accumulating the \textit{Estimated Values (\gls{EV}s)} of selected software requirement~\cite{harman_exact_2014}. However, several studies have argued that this assumption does not hold when interdependencies exist among requirements~\cite{carlshamre_release_2002,van_den_akker_flexible_2005,harman_exact_2014,mougouei2016factoring}. The reason is that software requirements affect each other's values due to the value-related dependencies among them~\cite{mougouei2020dependency,mougouei2019dependency,mougouei2018mathematical,mougouei2017dependency,mougouei2017modeling,dahlstedt_requirements_2005,carlshamre_release_2002,karlsson_improved_1997}. As a result, excluded requirements may impact the values of selected requirements that depend on them. Hence the actual value i.e. \textit{Customer Value (\gls{CV})} of a selected requirement might differ from its estimated value (EV).    


On the other hand, value-related requirement dependencies can be of various strengths in the context of real world projects~\cite{dahlstedt_moulding_2003,Robinson_RIM_2003,ramesh_toward_2001}. In other words, values of requirements can weakly, moderately, or strongly depend on each other~\cite{wang_simulation_2012}. Therefore, it is important to consider both the existence and the strengths of value-related requirement dependencies~\cite{brasil_multiobjective_2012,harman_exact_2014} while considering the impacts of requirements on each other's values during a requirement selection.


However, the existing requirement selection models either assume that requirements are independent~\cite{karlsson_optimizing_1997,jung_optimizing_1998,zhang_multi_objective_2007,finkelstein2009search,del2010ant} or they formulate requirement dependencies as precedence constraints of BKP formula without considering the strengths of those dependencies~\cite{bagnall_next_2001,li2016value,veerapen2015integer,greer_software_2004,ruhe_quantitative_2003,van2011quantitative,zhang2010search,saliu2007bi,van2005determination,ngo2009optimized,chen2013ant,xuan2012solving,van2008software}. In the latter case, dependencies are considered as binary (0/1) relations. 

As such, by excluding a requirement from an optimal set, all of its dependent requirements also have to be excluded even if the budget allows for their implementation~\cite{li_integrated_2010}. This problem is referred to as \textit{Selection Deficiency Problem (\gls{SDP})}~\cite{mougouei2016factoring}. As a result of the SDP, any small increase in the number of dependencies would dramatically depreciate the accumulated value of an optimal set~\cite{li_integrated_2010}. Hence, the SDP can severely impact the efficiency of selection models that formulate all requirement dependencies as precedence relations. 

This paper has focused on considering the impacts of value-related requirement dependencies on the value of an optimal set during a selection process. We have achieved this through considering both the existence and the strengths of value-related requirement dependencies during software requirement selection. In doing so, we have specially made three main contributions. 

First, we have demonstrated using algebraic structure of fuzzy graphs~\cite{kalampakas_fuzzy_2013,zimmermann_fuzzy_1996,rosenfeld_fuzzygraph_1975,mordeson_fuzzy_2001} to modeling value-related requirement dependencies and their strengths.

Second, we have proposed a new measure of value of an optimal set referred to as \textit{Overall Value (\gls{OV})} that considers the impacts of value-related requirement dependencies on the value of an optimal set (selected requirements). In this regard, we have introduced the following definitions of value.

\begin{itemize}
\item \textit{Estimated Value (EV):} Value of a software requirement estimated by the stakeholders.
\item \textit{Customer Value (CV):} Value of a software requirement derived by considering the impacts of value-related dependencies on the estimated value of that requirement.
\item \textit{Accumulated Value (AV):} a measure of value of an optimal set that is derived by accumulating the estimated values of selected requirements without considering value-related dependencies.
\item \textit{Overall Value (OV):} a measure of value of an optimal set that is derived by accumulating the customer values of selected requirements with the consideration of the impacts of value-related requirement dependencies on the value of the selected requirements.
\end{itemize}

Finally, we have presented a novel requirement selection model referred to as the \textit{\gls{GORS}} (Graph Oriented Requirement Selection), which maximizes the overall value of an optimal set while mitigating the SDP. The proposed model not only considers value-related dependencies among requirements but more importantly explicitly factors in the strengths of those dependencies. 

Validity and practicality of the work are verified through a) carrying out several simulations and b) studying a real world software project. The results of our simulations as well as a case study of a real world software project have consistently shown that: (a) the GORS model can properly capture the strengths of value-related dependencies among requirements during a selection process while mitigating the selection deficiency problem (SDP), (b) the GORS model always maximizes the overall value of an optimal set, and (c) maximizing the overall and the accumulated values of an optimal set can be conflicting objectives~\cite{korte_combinatorial_2006} as maximizing one may depreciate the other.

The remainder of this paper is organized as follows. We first discuss related works in Section~\ref{sec_related}. Then, Section~\ref{sec_modleing} gives the details of modeling value-related requirement dependencies by fuzzy graphs. After that, Section~\ref{sec_factoring} introduces our proposed formulation of overall value of an optimal set as well as our proposed graph oriented requirement selection (GORS) model. The results of our simulations and a case study of a real world software project are discussed in Section~\ref{sec_validation}. Section~\ref{sec_identification} then, discusses automated identification of value-related dependencies among requirements. Finally, Section~\ref{sec_conclusion} concludes the paper with a summary of major contributions and future work.    

%% file: related.tex
\section{Related Work}
\label{sec_related}

As discussed earlier, it is important to consider the impact of value-related dependencies during a requirement selection process~\cite{carlshamre_industrial_2001,carlshamre_release_2002,dahlstedt_requirements_2005}. On this basis, the existing requirement selection models can be categorized into two distinct groups. 

The first group of selection models~\cite{karlsson_optimizing_1997,jung_optimizing_1998,zhang_multi_objective_2007,finkelstein2009search,del2010ant} referred to as \textit{BKP models} are solely based on the basic formulation of binary knapsack problem~\cite{harman_exact_2014,zhang_multi_objective_2007,carlshamre_release_2002,lust_multiobjective_2012}
as given in (\ref{Eq_BKP}). 

Given a set of requirements \gls{R} $=\{r_1,...,r_n\}$, for each \gls{r_i} $\in R$, \gls{v_i} and \gls{c_i} in (\ref{Eq_BKP}) denote the estimated value and the cost of $r_i$ respectively. Also, \gls{b} denotes the available budget of the release. A binary variable $x_i$ specifies whether the requirement $r_i$ is selected (\gls{x_i} $=1$) or otherwise ($x_i=0$). As given by (\ref{Eq_BKP}), BKP models find a subset of $R$ that maximizes the accumulated value of selected requirements $(\sum_{i=1}^{n} v_i  x_i)$ without considering dependencies among them~\cite{ruhe_trade_2003,karlsson_optimizing_1997,jung_optimizing_1998,zhang_multi_objective_2007}.

\begin{equation}
\label{Eq_BKP}
\begin{aligned}
      \text{Maximize } & \sum_{i=1}^{n} v_i x_i   \\
      \text{Subject to} & \sum_{i=1}^{n} c_i x_i \leq b \\
			& x_i \in \{0,1\}
\end{aligned}
\end{equation}

To consider dependencies, the second group of selection models~\cite{bagnall_next_2001,li2016value,veerapen2015integer,greer_software_2004,ruhe_quantitative_2003,van2011quantitative,zhang2010search,saliu2007bi,van2005determination,ngo2009optimized,chen2013ant,xuan2012solving,van2008software} referred to as \textit{BKP-PC models} formulate dependencies as precedence constraints of the BKP formula as given by (\ref{Eq_BKP-PC}). For instance, consider the set of requirements $R=\{r_1,r_2,r_3,r_4\}$ for which we have the dependency set of \gls{D} $=\{(r_1,r_2),(r_2,r_3)\}$. 

The explicit dependency $(r_1,r_2)$ $\in D$ (when formulated as a precedence relation) means that requirement $r_2$ cannot be included in the optimal set unless $r_3$ is also selected. Equation (\ref{Eq_BKP-PC}) formulates these dependencies as precedence constraints set $PCS = \{x_2 \leq x_3, x_1 \leq x_2\}$. Moreover, the explicit dependencies $(r_1,r_2)$ and $(r_2,r_3)$ can also imply an implicit dependency from $r_1$ to $r_3$. 

\begin{equation}
\label{Eq_BKP-PC}
     \begin{aligned}
      \text{Maximize } & \sum_{i=1}^{n} v_i x_i   \\
      \text{Subject to} &\sum_{i=1}^{n} c_i  x_i \leq b \\
			& x_i \in \{0,1\}\\
	& x_i \leq x_j,\text{ if } r_i \text{ depends on } r_j 
     \end{aligned}
\end{equation}

However, BKP-PC models merely captures the existence of dependencies while ignoring their strengths. As such, all dependencies are treated as binary (0/1) relations and consequently a requirement cannot be selected even in the presence of sufficient budget unless all of its dependent requirements are also selected. This makes BKP-PC models prone to the selection deficiency problem (SDP)~\cite{mougouei2016factoring} as discussed earlier.  

As a result of the SDP, any increase in the number of dependencies would dramatically depreciate the accumulated value of selected requirements~\cite{li_integrated_2010}. Therefore the SDP can severely impact the efficiency of BKP-PC models. In one study, Chen et al.~\cite{li_integrated_2010} demonstrated that a $2\%$ increase in the number of dependencies (when formulated as precedence relations) would lead to almost $10\%$ decrease in the accumulated value of the optimal set. 

The SDP occurs if the condition of (\ref{Eq_SDP}) holds. The dependency set $D$ in (\ref{Eq_SDP}) specifies the explicit dependencies among a set of requirements $R=\{r_1,...,r_n\}$ where $R$ is partitioned into two distinct subsets: an optimal set \gls{O} $\subseteq R $ (selected requirements) and an excluded set \gls{Otilde} $\subseteq R$ (ignored requirements) such that $O \cap \tilde{O} = \emptyset$. 

\begin{align}
\label{Eq_SDP}
\exists\text{ } r_i,r_j\in \tilde{O} : (r_i,r_j) \in D,& (\sum_{r_{k}\in O} c_k) + c_i \leq b \\ \nonumber
& (\sum_{r_{k}\in O} c_k) + c_i  + c_j > b 
\end{align}

Akker et al.~\cite{van_den_akker_flexible_2005} took a different approach toward requirement selection by first partitioning requirements into a set of $m \geq 1$ distinct collections $S=\{s_1,...,s_m\}$ and then assigning a value to each individual collection. As given by (\ref{Eq_BKP-CS}), for each realized collection $s_j\in S$, the difference between value of collection ($w_j$) and its corresponding accumulated value $\sum_{r_k\in s_j}v_k$ is added to the accumulated value of selected requirements ($\sum_{i=1}^{n} v_i$) to derive the value of the optimal set. $y_j$ in (\ref{Eq_BKP-CS}) specifies whether collection $s_j$ is fully realized ($y_j=1$) or otherwise ($y_j=0$). However, this model does not specify how to find the values of collections. Furthermore, estimating the values of the collections can become computationally as complex as $O(2^n)$ for $n$ requirements. 

Some of the existing works~\cite{li_integrated_2010,sagrado_multi_objective_2013,Zhang_RIM_2013} have addressed the latter problem by only estimating the values of collections of size $2$ using pairwise comparisons. Pairwise comparisons however, cannot capture implicit value dependencies among requirements. Furthermore, estimating the values of collections (pairs) still remains subjective. 


\begin{equation}
\label{Eq_BKP-CS}
\begin{aligned}
      \text{Maximize } &\sum_{i=1}^{n} v_i x_i + \sum_{j=1}^{m}(w_j-\sum_{r_k\in s_j}v_k)  y_j\\
      \text{Subject to} &\sum_{i=1}^{n} c_i  x_i \leq b \\
			       & x_i,y_j \in \{0,1\} \\
\end{aligned}
\end{equation}

%% file: modeling.tex
\section{Modeling Value-related Requirement Dependencies using Fuzzy Graphs}
\label{sec_modleing}
This section highlights our reasons for choosing fuzzy graphs and then gives the details of employing fuzzy graphs for modeling value-related requirement dependencies.

%% file: modeling_whyfuzzygraph.tex
\subsection{Why Fuzzy Graphs?}
\label{sec_fuzzygrpahs_why}

Requirement dependencies in general and value-related requirement dependencies in particular are fuzzy relations~\cite{carlshamre_industrial_2001} in the sense that strengths of those dependencies vary from large to insignificant in the context of real world software projects~\cite{carlshamre_industrial_2001,AN_fuzzyDependency_2005,wang_simulation_2012}.


Fuzzy graphs on the other hand, have demonstrated to contribute to more accurate models in computer science and engineering~\cite{Mathew_strong_2013,mordeson_applications_2000} by considering uncertainty~\cite{mougouei2013fuzzy,mougouei2012measuring,mougouei2012goal} in real world problems~\cite{zadeh_fuzzysets_1965,Klir_1987_Uncertainty}. 

Hence, fuzzy graphs can properly capture the fuzziness associated with value-related dependencies among software requirements and contribute to more accurate models for software requirement selection.

It is also worth mentioning that value-related requirement dependencies are directed relations. In other words, $(r_i,r_j) \not \Rightarrow (r_j,r_i)$. That is customer value of a requirement $r_i$ depends on a requirement $r_j$, does not imply that the customer value of $r_j$ also depends on $r_i$. As such, we adopt directed fuzzy graphs~\cite{zimmermann_fuzzy_1996} for modeling value-related requirement dependencies.

%% file: modeling_interdependency.tex
\subsection{Fuzzy Requirement Interdependency Graphs}
\label{sec_fuzzygrpahs_building}

Based on the definition of fuzzy graphs~\cite{rosenfeld_fuzzygraph_1975}, a \textit{Fuzzy Requirement Interdependency Graph (\gls{FRIG})} is defined as a directed fuzzy graph \gls{G} $=(R,D,$ \gls{mu}$,$ \gls{rho}$)$ in which a non-empty set of identified requirements $R=\{r_1,...,r_n\}$ constitute the graph nodes and a set of (explicit) dependencies among the requirements $D=R\times R$ form the edges of the graph. 

A dependency $(r_i,r_j) \in D$ means that the customer value of $r_i$ explicitly depends on selection of $r_j$. The membership function $\rho: R\times R\rightarrow [0,1]$ denotes the strengths of explicit value-related dependencies (membership degrees of edges) in $D$. $\rho(x,y)=0$ denotes the absence of an explicit dependency from $x$ to $y$.

The fuzzy membership function $\mu$ specifies the membership degree of requirements in $R$. Requirements of a software are either identified and listed in its requirement set $R$ or they are unidentified. Hence, we have $\forall r_i \in R : \mu(r_i)=1$. Therefore, $G=(R,D,\mu,\rho)$ can be abbreviated as $G=(R,D,\rho)$.

On the other hand, for $G=(R,D,\mu,\rho)$ to be a fuzzy graph, the following condition must hold at all times. $\forall(x,y) \in D : \rho(x,y)\leq \mu(x)\wedge\mu(y)$, where \gls{wedge} denotes fuzzy AND operator (taking infimum). In other words, for $\rho(x,y)$ to denote a fuzzy relation, the membership degree of a relation, also referred to as the strength of the relation (dependence) must not exceed the membership degree of either of the two elements. Theorem~(\ref{Theorem_1}) shows that a FRIG always satisfies the condition of fuzzy graphs.  
\begin{mydef}
\label{Theorem_1}
If $G=(R,D,\rho)$ is a FRIG, ($\forall r_i \in R : \mu(r_i)=1$) then $G$ always satisfies the condition $\forall(r_i,r_j) \in D : \rho(r_i,r_j)\leq \mu(r_i)\wedge\mu(r_j)$.
\end{mydef}
\begin{proof}
$G=(R,D,\rho)$ is a FRIG $\Rightarrow$ \\
\vspace{-.2cm}
\addtolength{\itemsep}{-0.4\baselineskip}
	\item $a)\text{ } \forall r_i \in R : \mu(r_i)=1 \Rightarrow \forall(r_i,r_j) \in D : \mu(r_i)\wedge\mu(r_j) = infimum(\mu(r_i),\mu(r_j)) = 1$,\\
	\item $b) \text{ } \rho: R \times R \rightarrow [0,1] \Rightarrow \forall r_i,r_j \in R, \rho (r_i,r_j) \leq 1$.\\
Therefore, $\forall(r_i,r_j) \in D : \rho(r_i,r_j)\leq \mu(r_i)\wedge\mu(r_j)$.
\end{proof}
\begin{exmp}
\label{ex_fuzzy_relation}
Let $E_1=(R,D,\rho)$ be a FRIG with $R=\{r_1,r_2,r_3,r_4\}$, $D=\{(r_1,r_2), $$(r_2,r_3)$$, $$(r_3,r_4)$$,(r_4,r_2)\}$ as in Figure~\ref{fig_ex_fuzzy_relation}. The membership function $\rho$ specifies the strengths of explicit dependencies in $D$ as $\rho(r_1,r_2)$$=0.6$, $\rho(r_2,r_3)=0.4,$ $\rho(r_3,r_4)=0.8,$ $\rho(r_4,r_2)=0.2$. The dependency $(r_1,r_2)$ specifies that customer value of $r_1$ explicitly depends on $r_2$ and $\rho(r_1,r_2)$ gives the strength of the dependency ($0.6$). 
\end{exmp}
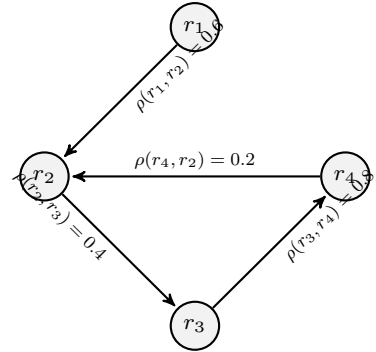
\begin{figure}[h!]
\centering
\input{fig_ex_fuzzy_relation}
\caption{FRIG of Example~\ref{ex_fuzzy_relation}}
\label{fig_ex_fuzzy_relation}
\end{figure}

Value-related dependencies in a FRIG can be either explicit or implicit. Explicit dependencies are identified by the edges of the graph whereas implicit dependencies are inferred from explicit dependencies. For instance, an implicit dependency $(r_1,r_2,r_3)$ from $r_1$ to $r_3$ in Figure~\ref{fig_ex_fuzzy_relation} is inferred from explicit dependencies $(r_1,r_2)$ and $(r_2,r_3)$. 

\begin{defn}
\label{fuzzygraph_def_4}
\textit{Value-related Dependencies, and their Strengths}. Let $P=\{p_1,p_2,..., p_m\}$ be the set of all value-related dependencies from node $r_0$ to node $r_n$ in a FRIG $G=(R,D,\rho)$. A value-related dependency $p_i \in P$ is defined as a sequence of distinct nodes ($r_0,...,r_n$) such that $\rho(r_{i-1},r_i) > 0$, where $1 \leq i \leq n$ and $n\geq 1$ is the length of the dependency. 

The strength of a value-related dependency $p_i\in P$ is derived by (\ref{Eq_fuzzygraph_strength_path}) that is the strength of the $p_i$ equals to the strength of the weakest explicit value-related dependency (edge) in $p_i$.
\end{defn}
\begin{eqnarray}
\label{Eq_fuzzygraph_strength_path}
\forall p_i=(r_0,...,r_n) \in P,\phantom{s}\rho(p_i) = \displaystyle\bigwedge_{j=1}^{n}\text{ }\rho(r_{j-1},r_j)
\end{eqnarray}
\vspace{0.2em}

In a FRIG $G=(R,D,\rho)$ with $m$ dependencies from a requirement $r_0$ to $r_n$, the overall strength of all dependencies from $r_0$ to $r_n$ is denoted as $\rho^{\infty}(r_0,r_n)$ and calculated by (\ref{Eq_fuzzygraph_strength_relation}). Based on (\ref{Eq_fuzzygraph_strength_relation}), the overall strength of all dependencies from $r_0$ to $r_n$ equals to the strength of the strongest dependency among all the $m$ dependencies from $r_0$ to $r_n$. It is clear that $\rho^{\infty}(r_0,r_n)=\rho(r_0,r_n)$ when there is no implicit value-related dependency from $r_0$ to $r_n$. 


\begin{align}
\label{Eq_fuzzygraph_strength_relation}
\phantom{ssss}\rho^{\infty}(r_0,r_n) = \displaystyle \bigvee_{i=1}^{m}\rho(p_i)
\end{align}

To measure the level of value-related dependencies in a FRIG $G=(R,D,\rho)$ with $n$ requirements ($|R|=n$) and $k$ explicit value-related dependencies among those requirements ($|D|=k$), we define \textit{Level Of Interdependency (\gls{LOI})} as given in (\ref{Eq_loi}). 

\begin{align}
\label{Eq_loi}
\phantom{sss}LOI(G)=\frac{k}{\perm{n}{2}}, \perm{n}{2}=\frac{n!}{(n-2)!}
\end{align}

\begin{exmp}
\label{ex_fuzzy_loi}
For the FRIG $E_1$ in Example~\ref{ex_fuzzy_relation}, with $n=4$, $k=4$ we have $LOI(E_1)=\displaystyle \frac{4}{\perm{4}{2}} = \frac{4}{12} =0.33$.
\end{exmp}

It is also worth mentioning that requirements of software projects may negatively influence the customer values of each other. For instance, a negative dependency from a requirement $r_i$ to a requirement $r_j$ means that the customer value of $r_i$ will be depreciated if $r_j$ is selected with $r_i$ in an optimal set. Such negative dependency nonetheless, can be modeled as a positive dependency from $r_i$ to $\bar{r_j}$ where $\bar{r_j}$ denotes ignoring $r_j$ (excluding $r_j$ from an optimal set). Hence, FRIGs can capture both positive and negative value-related dependencies. Nevertheless, in this paper we only focus on positive dependencies for the sake of simplicity.     

%% file: fig_ex_fuzzy_relation.tex
\begin{tikzpicture}[->,>=stealth',shorten >=1pt,auto,node distance=2.8cm,
  thick,main node/.style={circle,fill=gray!10,draw}]
  \node[main node] (1) {$r_1$};
  \node[main node] (2) [below left of=1] {$r_2$};
  \node[main node] (3) [below right of=2] {$r_3$};
  \node[main node] (4) [below right of=1] {$r_4$};

  \path[every node/.style={font=\sffamily\small,scale=.8}]
    (1) edge node [auto=left,sloped,pos=0.5] {$\rho(r_1,r_2)=0.6$} (2)
    (2) edge node [auto=right,sloped,pos=0.5] {$\rho(r_2,r_3)=0.4$} (3)
    (3) edge node [auto=right,sloped,pos=0.5] {$\rho(r_3,r_4)=0.8$} (4)
    (4) edge node [auto=right,pos=0.5] {$\rho(r_4,r_2)=0.2$} (2);
\end{tikzpicture}

%% file: factoring.tex
\section{Optimizing the Overall Value of an Optimal Set using the GORS Model}
\label{sec_factoring}

This section gives the details of considering value-related requirement dependencies in calculating the overall value of an optimal set. We then present our proposed graph oriented requirement selection model referred to as the GORS model that optimizes the overall value an optimal.

%% file: factoring_ov.tex
\subsection{Overall Value of an Optimal Set}

During a selection process some of the requirements of a software may be excluded from the optimal set. Due to the value-related dependencies among requirements however, excluded requirements may impact the values of selected requirements that depend on them. Equation (\ref{Eq_impact}) captures these impacts. For a FRIG $G=(R,D,\rho)$, $O=\{o_1,...,o_m\}$ and $\tilde{O}=\{\tilde{o}_1 ,...,\tilde{o}_k\}$ denote selected and excluded requirements respectively such that $O \subseteq R , \tilde{O} \subseteq R : O \cap \tilde{O} = \emptyset, O \cup \tilde{O} = R$. 

For $\forall o_i \in O$, \gls{I_i} indicates the impact of excluded requirements on $o_i$'s value. This impact is calculated by taking supremum (fuzzy OR operator \gls{vee}) over the strengths of all dependencies from $o_i$ to the excluded requirements in $\tilde{O}$. For an excluded requirement $r_j \in \tilde{O}$ then, the overall strength of all dependencies from $o_i$ to $\tilde{o_j}$ is denoted by $\rho^{\infty}(o_i,\tilde{o_j})$ which specifies the extent to which $o_i$'s value relies on selection of $\tilde{o_j}$ (through all dependencies from $o_i$ to $o_j$) as derived by (\ref{Eq_fuzzygraph_strength_relation}). 

\begin{eqnarray}
\label{Eq_impact}
&\hspace{1em}\forall r_i \in R: I_{i}= 
	 \begin{cases}
			\displaystyle\bigvee  ^k_{j=1}(\rho^{\infty}(r_i,\tilde{o}_j)) & \text{if $r_i \in O$} \\
			0 & \text{if $r_i \notin O$}
	\end{cases}
\end{eqnarray}   

As discussed earlier, accumulated value (AV) of an optimal set $O$ is derived by accumulating the estimated values of selected requirements ($\sum_{r_i \in O} v_i$). The overall value (OV) of $O$ in contrast, is derived by accumulating the customer values of selected requirements as computed by (\ref{Eq_value}). Overall value of $O$ ($\sum_{r_i \in O} CV_i$) captures the impacts of value-related dependencies as customer value of each requirement $r_i \in O$ ($CV_i$) captures the impacts of value-related dependencies on the value of $r_i$ as given by~(\ref{Eq_cv}).

\begin{align}
\label{Eq_cv}
\hspace{1em}CV_i = v_i\times(1-I_{i})
\end{align}

\vspace{-1em}
\begin{equation}
\label{Eq_value}
\begin{aligned}
\hspace{1em}OV = \sum_{r_i \in O} CV_i = \sum_{r_i \in O} v_i\times(1-I_{i})
\end{aligned}
\end{equation}

%% file: factoring_gors.tex
\subsection{The \textit{GORS} Selection Model}
\label{sec_factoring_gors}
In contrast to BKP and BKP-PC models that aim to maximize the accumulated value of an optimal set, the proposed GORS model, maximizes the overall value of an optimal set by factoring in the impacts of excluded requirements on the values of selected requirements as given by (\ref{Eq_GORS}). 

\begin{equation}
\label{Eq_GORS}
     \begin{aligned}
      \text{Maximize } & \sum_{i=1}^{n} v_i x_i (1-I_{i})\\
      \text{Subject to} & \sum_{i=1}^{n} c_i x_i < b \\
       & x_i = \{0,1\}
     \end{aligned}
\end{equation}

%% file: factoring_examples.tex
\subsection{Examples of Requirement Selection}
\label{sec_factoring_examples}

This section provides examples of requirement selection using BKP, BKP-PC, and GORS models. 

\begin{exmp}
	\label{ex_selection}
	Let $G_p=(R,D,\rho)$ be a FRIG of a software (Figure~\ref{fig_ex_selection}) with requirement set $R=\{r_1,r_2,r_3,r_4\}$ and explicit value-related dependencies $D$ as in Figure~\ref{fig_ex_selection} with strengths of $\rho(r_1,r_2)= 0.4$, $\rho(r_1,r_3)= 0.8$, $\rho(r_2,r_4)= 0.3$, $\rho(r_3,r_1)= 0.8$, $\rho(r_3,r_2)= 0.6$, $\rho(r_3,r_4)= 0.8$, and $\rho(r_4,r_3)= 0.2$. The costs and values of the requirements are specified by $C=\{c_1=10,c_2=10,c_3=15,c_4=10\}$ and $V=\{v_1=20,v_2=10,v_3=50,v_4=10\}$ respectively. 
\end{exmp}

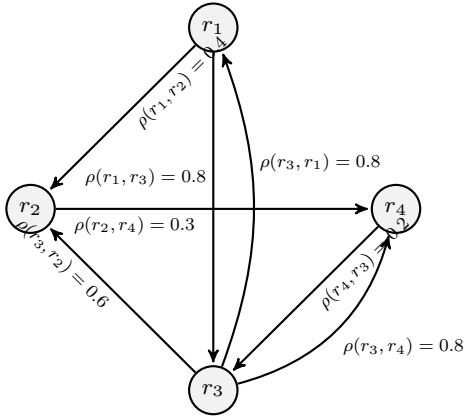
\begin{figure}[h!]
\centering
\input{fig_ex_selection}
\caption{FRIG of Example~\ref{ex_selection}}
\label{fig_ex_selection}
\end{figure}
\begin{table}[!h]
\renewcommand{\arraystretch}{1.3}
\caption{Overall strengths of the value-related dependencies among requirements of Example~\ref{ex_selection}}
\label{table_ex_strengths}
\centering
\input{table_ex_strengths}

\end{table}
\begin{table}[!h]
\renewcommand{\arraystretch}{1.3}
\caption{Accumulated values, overall values, and accumulated costs of the requirement subsets of Example~\ref{ex_selection}}
\label{table_ex_all}
\centering
\input{table_ex_all}

\end{table}

\begin{exmp}
\label{ex_BKP}
Consider finding the optimal set of requirements by the BKP model. Among all subsets of $R$ in Table~\ref{table_ex_all}, the BKP model recommends $s_6=\{r_1,r_3\}$ as the optimal set with the highest accumulated value of $AV=70$ and the accumulated cost of $AC=25$. Therefore we have $O=\{r_1,r_3\}$, $\tilde{O}=\{r_2,r_4\}$. In order to compute the overall value of the optimal set, we first calculate impacts of excluded requirements on the values of selected requirements based on (\ref{Eq_impact}). Impacts are calculated based on overall strengths of value-related dependencies in Table~\ref{table_ex_strengths}. 

Overall strengths of dependencies are calculated by (\ref{Eq_fuzzygraph_strength_relation}) as explained earlier. For instance, to compute the overall strength of the dependency from $r_4$ to $r_2$, dependencies $p_{1}=(r_4,r_3,r_2)$ and $p_2=(r_4,r_3,r_1,r_2)$ need to be considered. Based on (\ref{Eq_fuzzygraph_strength_relation}), supremum is taken over the strengths of the two dependencies to calculate the overall strength of the dependency: $\rho^{\infty}(r_4,r_2)=\vee((\rho(r_4,r_3)\wedge\rho(r_3,r_2)),(\rho(r_4,r_3)\wedge\rho(r_3,r_1)\wedge\rho(r_1,r_2))) = \vee((0.2\wedge0.6),(0.2\wedge0.8\wedge0.4)) = \vee(0.2,0.2) = 0.2$. 

Impacts then can be computed as $I_{1} = \vee(\rho^{\infty}(r_1,r_2)$, $\rho^{\infty}(r_1,r_4)) = 0.8$, $I_{3} = \vee(\rho^{\infty}(r_3,r_2),\rho^{\infty}(r_3,r_4)) = 0.8$. Finally, overall value of the optimal set $O=\{r_1,r_3\}$ is calculated as $OV = v_{1} \times (1-I_{1}) + v_{3} \times (1-I_{3}) = 20 \times 0.2 + 50 \times 0.2 = 14$ which is less than the overall value of $s_{10}$. Therefore, the BKP model does not necessarily maximize the overall value of an optimal set. 
\end{exmp}

\begin{exmp}
\label{ex_BKP_Pre}
Consider finding the optimal set of requirements in Example~\ref{ex_BKP} using the BKP-PC model, which finds a subset of requirements with the highest AV respecting the budget ($AC \leq 25$) and the precedence constraints among the requirements. We can derive a precedence constraints set $PCS=\{x_1 \leq x_2$, $x_1 \leq x_3$, $x_2 \leq x_4$, $x_3 \leq x_1$, $x_3 \leq x_2$, $x_3 \leq x_4$, $x_4 \leq x_3\}$ from the dependency set $D$ of $G_p$. Obviously there is only one case that simultaneously satisfies both $PCS$ and $AC \leq 25$, which is the empty set $s_0=\{\}$ ($x_1=x_2=x_3=x_4 = 0$) in Table~\ref{table_ex_all} with $AV=OV=0$. Hence, none of the requirements can be implemented even though the budget is available for implementing some of them. This is due to the selection deficiency problem (SDP) as discussed before.
\end{exmp}

\begin{exmp}
\label{ex_GORS}
Consider finding the optimal set of requirements in Example~\ref{ex_BKP} using the GORS model, which finds a subset of requirements with the highest OV respecting $AC \leq 25$. To do so, we first calculate the OV of all subsets of $R$ (steps of calculation was demonstrated for $s_6$ in Example~\ref{ex_BKP}) as listed in Table~\ref{table_ex_all}. Among all subsets of the requirements, $s_{10}=\{r_3,r_4\}$ gives the highest overall value of $OV=18$ while the accumulated cost is within the budget, that is $(AC=25 \leq 25)$. Therefore, $s_{10}=\{r_3,r_4\}$ will be selected as the optimal set. $s_{10}$ however, is not giving the maximum accumulated value. $s_6$ for instance, provides a higher AV.
\end{exmp}

%% file: fig_ex_selection.tex
\begin{tikzpicture}[->,>=stealth',shorten >=1pt,auto,node distance=3.4cm,
  thick,main node/.style={circle,fill=gray!10,draw}]
  \node[main node] (1) {$r_1$};
  \node[main node] (2) [below left of=1] {$r_2$};
  \node[main node] (3) [below right of=2] {$r_3$};
  \node[main node] (4) [below right of=1] {$r_4$};
  \path[every node/.style={font=\sffamily\small,scale=0.8}]
		(3) edge [bend right = 20] node [auto=right,pos=0.6] {$\rho(r_3,r_1)=0.8$} (1)
    (3) edge node [auto=righ,sloped,pos=0.5] {$\rho(r_3,r_2)=0.6$} (2)
		(3) edge [bend right] node [auto=right,pos=0.5] {$\rho(r_3,r_4)=0.8$} (4)
    (1) edge node [auto=left,sloped,pos=0.5] {$\rho(r_1,r_2)=0.4$} (2)
		(1) edge node [auto=right,pos=0.4] {$\rho(r_1,r_3)=0.8$} (3)
    (2) edge [right] node [auto=right,pos=0.25] {$\rho(r_2,r_4)=0.3$}(4)
    (4) edge node [auto=left,sloped,pos=0.5] {$\rho(r_4,r_3)=0.2$} (3);
\end{tikzpicture}

%% file: table_ex_strengths.tex
\begin{tabular}{ccccc}
  \toprule[1.5pt]
  \head{$\rho^{\infty}(x,y)$} & \head{$r_1$} & \head{$r_2$} & \head{$r_3$} & \head{$r_4$}\\
  \midrule
  $r_1$ & $1.0$ & $0.6$ & $0.8$ & $0.8$\\
  $r_2$ & $0.2$ & $1.0$ & $0.2$ & $0.3$\\
	$r_3$ & $0.8$ & $0.6$ & $1.0$ & $0.8$\\
  $r_4$ & $0.2$ & $0.2$ & $0.2$ & $1.0$\\
  \bottomrule[1.5pt]
\end{tabular}

%% file: table_ex_all.tex
\resizebox {.36\textwidth }{!}{\begin{tabular}{llll}
  \toprule[1.5pt]
  \head{$Subset$} & \head{$AC$} & \head{$AV$} & \head{$OV$}\\
  \midrule
	$s_0= \{\}$    & $0$  & $0$  & $0$ \\
    $s_1= \{r_1\}$ & $10$ & $20$ & $4$ \\
	$s_2= \{r_2\}$ & $10$ & $10$ & $7$ \\
    $s_3= \{r_3\}$ & $15$ & $50$ & $10$ \\
    $s_4= \{r_4\}$ & $10$ & $10$ & $8$\\
	$s_5=\{r_1,r_2\}$ & $20$ & $30$ & $11$\\
    $s_6=\{r_1,r_3\}$ & $25$ & $70$ & $14$\\
    $s_7=\{r_1,r_4\}$ & $20$ & $30$ & $12$\\
    $s_8=\{r_2,r_3\}$  & $25$ & $60$ & $17$\\
    $s_9=\{r_2,r_4\}$  & $20$ & $20$ & $16$\\
    $s_{10}=\{r_3,r_4\}$ & $25$ & $60$ & $18$\\
    $s_{11}=\{r_1,r_2,r_3\}$  & $35$ & $80$ & $21$\\
	$s_{12}=\{r_1,r_2,r_4\}$  & $30$ & $40$ & $20$\\
    $s_{13}=\{r_1,r_3,r_4\}$  & $35$ & $80$ & $36$\\
    $s_{14}=\{r_2,r_3,r_4\}$  & $35$ & $70$ & $26$\\
    $s_{15}=\{r_1,r_2,r_3,r_4\}$ & $45$ & $90$ & $90$\\
  \bottomrule[1.5pt]
\end{tabular}}

%% file: validation.tex
\section{Validation}
\label{sec_validation}

Validity and practicality of our work are verified through carrying out several simulations and studying a real world software project. 

%% file: simulation.tex
\subsection{Simulations}
\label{sec_simulation}


%% file: simulation_design.tex
\subsubsection{Simulation Design}
\label{sec_simulation_design}

We compared performance of the GORS model against those of the BKP and BKP-PC models through carrying out simulations on requirements from two classic requirement sets~\cite{Karlsson_1996,karlsson_optimizing_1997,jung_optimizing_1998} from real world projects of Ericssons \textit{Radio Access Network (\gls{RAN})} and \textit{Performance Management Traffic Recording (\gls{PMR})} with $14$ and $11$ requirements respectively. Estimated values and costs of the requirements of the RAN and PMR projects are listed in Table~\ref{table_data}. 

\begin{table}[!h]
\renewcommand{\arraystretch}{1.3}
\caption{Estimated values and costs of requirements for the RAN and PMR projects}
\label{table_data}
\centering
\input{table_data}

\end{table}

Simulation has been widely used for the purpose of evaluation in system analysis~\cite{Law_1997_SMA} as well as the studies concerning requirement dependencies~\cite{wang_simulation_2012,li_integrated_2010}. Chen et al.~\cite{li_integrated_2010} for instance, proposed simulating requirement dependencies for analyzing the performance of requirement selection models. 

However, to the best of our knowledge there is no work in the existing literature to study distribution of strengths of dependencies among software requirements. Hence, we simulate the strengths of explicit value-related requirement dependencies with uniformly distributed random numbers in $[0,1]$ generated by the \textit{nextDouble()} \textit{Method} of the \textit{Class Random} in Java~\cite{java_random}. Our aim is to allow for evaluating the performance of selection models for various levels of dependencies among requirements. 


Our simulation process starts with construction of a fuzzy requirement interdependency graph (FRIG) with randomly generated strengths of explicit value-related dependencies (edges of the graph) for a given level of interdependency ($LOI \in [0,1]$). A range of budgets ($Budget=\{1,2, ..., 120\}$) then will be specified to examine the performance of the selection models in the presence of various budget constraints. 

At the end of each simulation, an optimal set of requirements will be generated by each of the selection models. Then, the accumulated value and the overall value of each optimal set will be calculated and compared against those of the other selection models. The simulation will be repeated for different levels of interdependency (LOIs) among requirements. 

%% file: table_data.tex
\resizebox {.41\textwidth }{!}{\begin{tabular}{ccccccc}
  \toprule[1.5pt] 
    Requirements & \multicolumn{2}{c}{RAN} & \multicolumn{2}{c}{PMR}\\
    & \head{Value} & \head{Cost} & \head{Value} & \head{Cost} \\
    \midrule
    $r_1$     &$12$ &$1$  &$0$  &$6$  \\
		$r_2$     &$6$  &$2$  &$6$  &$5$  \\
    $r_3$     &$5$  &$3$  &$3$  &$6$  \\
    $r_4$     &$7$  &$4$  &$11$ &$19$  \\
    $r_5$     &$12$ &$6$  &$32$ &$28$  \\
		$r_6$     &$16$ &$11$ &$20$ &$4$  \\
    $r_7$     &$3$  &$4$  &$9$  &$5$  \\
    $r_8$     &$3$  &$6$  &$4$  &$7$  \\
    $r_9$     &$4$  &$7$  &$25$ &$10$  \\
    $r_{10}$  &$5$  &$12$ &$9$  &$3$  \\
		$r_{11}$  &$1$  &$4$  &$3$  &$8$  \\
    $r_{12}$  &$1$  &$6$  &$-$  &$-$  \\
    $r_{13}$  &$21$ &$23$ &$-$  &$-$  \\
    $r_{14}$  &$3$  &$10$ &$-$  &$-$  \\
    \bottomrule[1.5pt]
\end{tabular}
}

%% file: simulation_result.tex
\subsubsection{Simulation Results}
\label{sec_simulation_result}


Figure~\ref{fig_result_all} shows the results of our simulations. The $x$ and $y$ axes show the available budget ($Budget\in \{1,...,120\}$) and the level of interdependency ($LOI \in \{0,0.1,...,1\}$) respectively. The $z$ axis shows the percentage of the accumulated value (overall value) of the optimal set, which is the ratio of AV (OV) to the total estimated value of the requirements multiplied by $100$. 

Our simulation results consistently showed that the BKP model maximized the accumulated value (AV) while the GORS model maximized the overall value (OV) of optimal sets. Nevertheless, none of these models simultaneously maximized both AV and OV for an optimal set. In other words, maximizing AV and OV demonstrated to be conflicting objectives.

\begin{figure*}
	\begin{center}
		\hspace{0em}	
		\subfigure[Accumulated value of RAN]{%
			\label{fig:first}
			\includegraphics[scale=0.5]{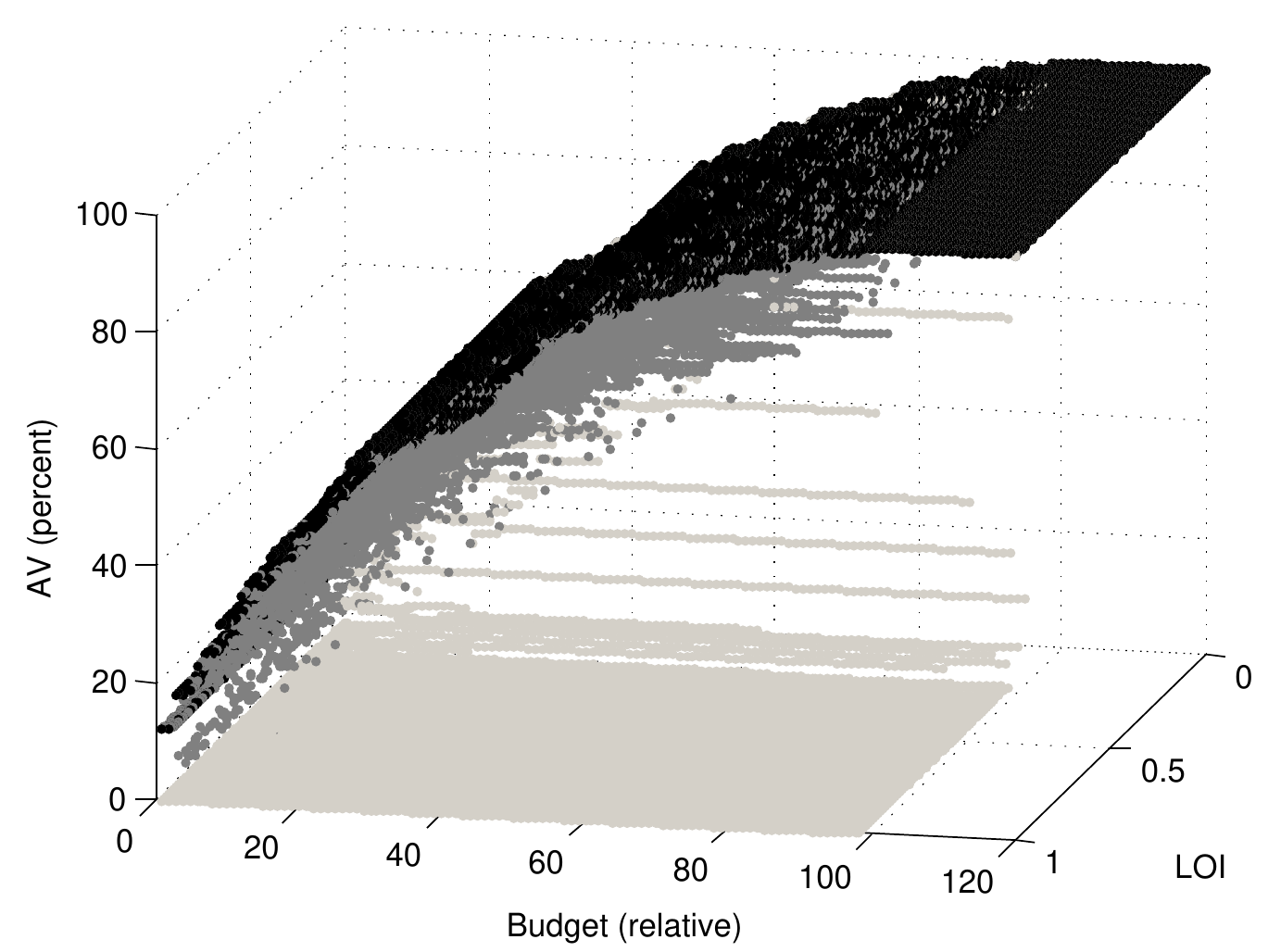}
		}%
		\subfigure[Accumulated value of PMR]{%
			\label{fig:second}
			\includegraphics[scale=0.5]{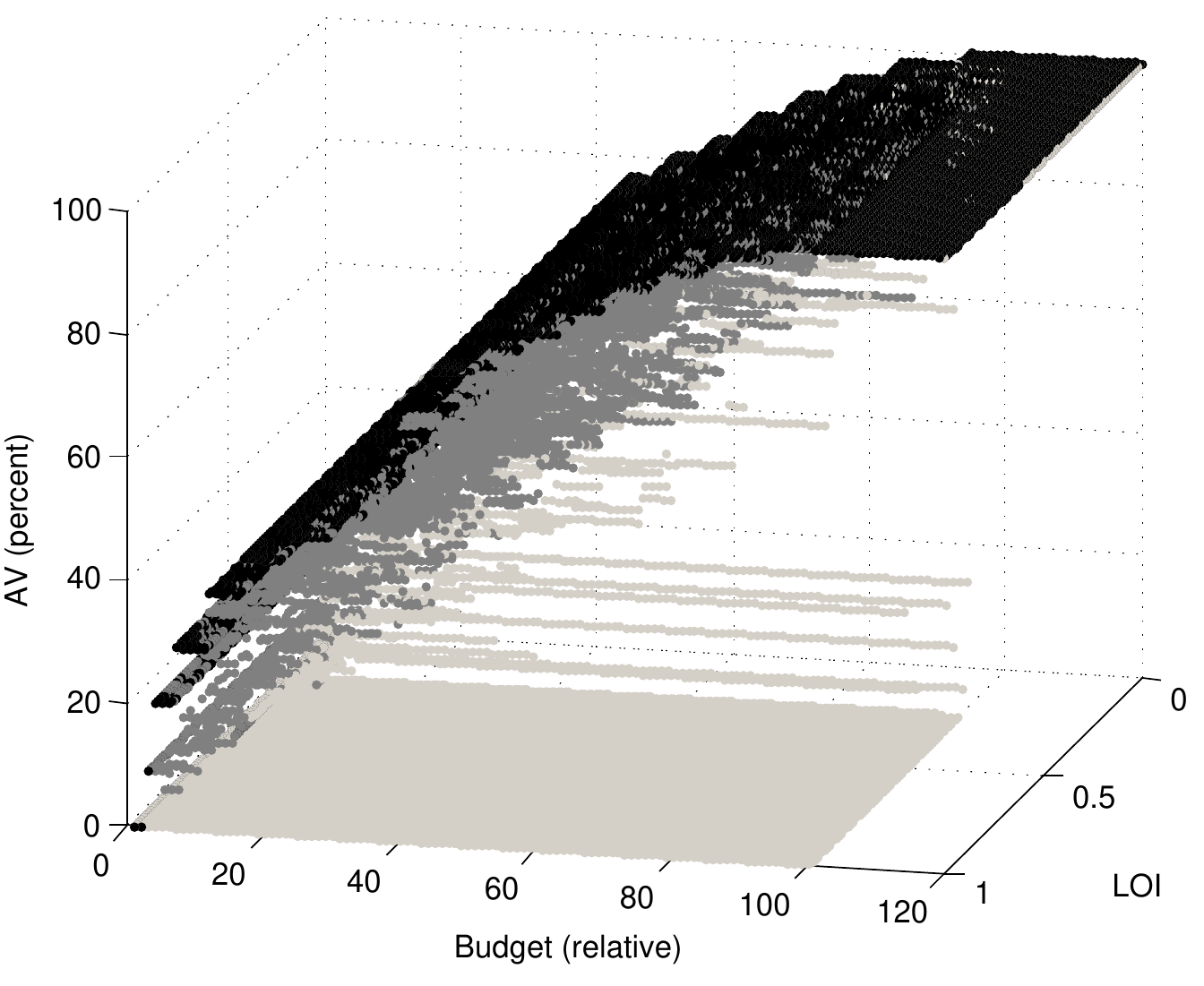}
		}\\ 
		\subfigure[Overall Value of RAN]{%
			\label{fig:second}
			\includegraphics[scale=0.5]{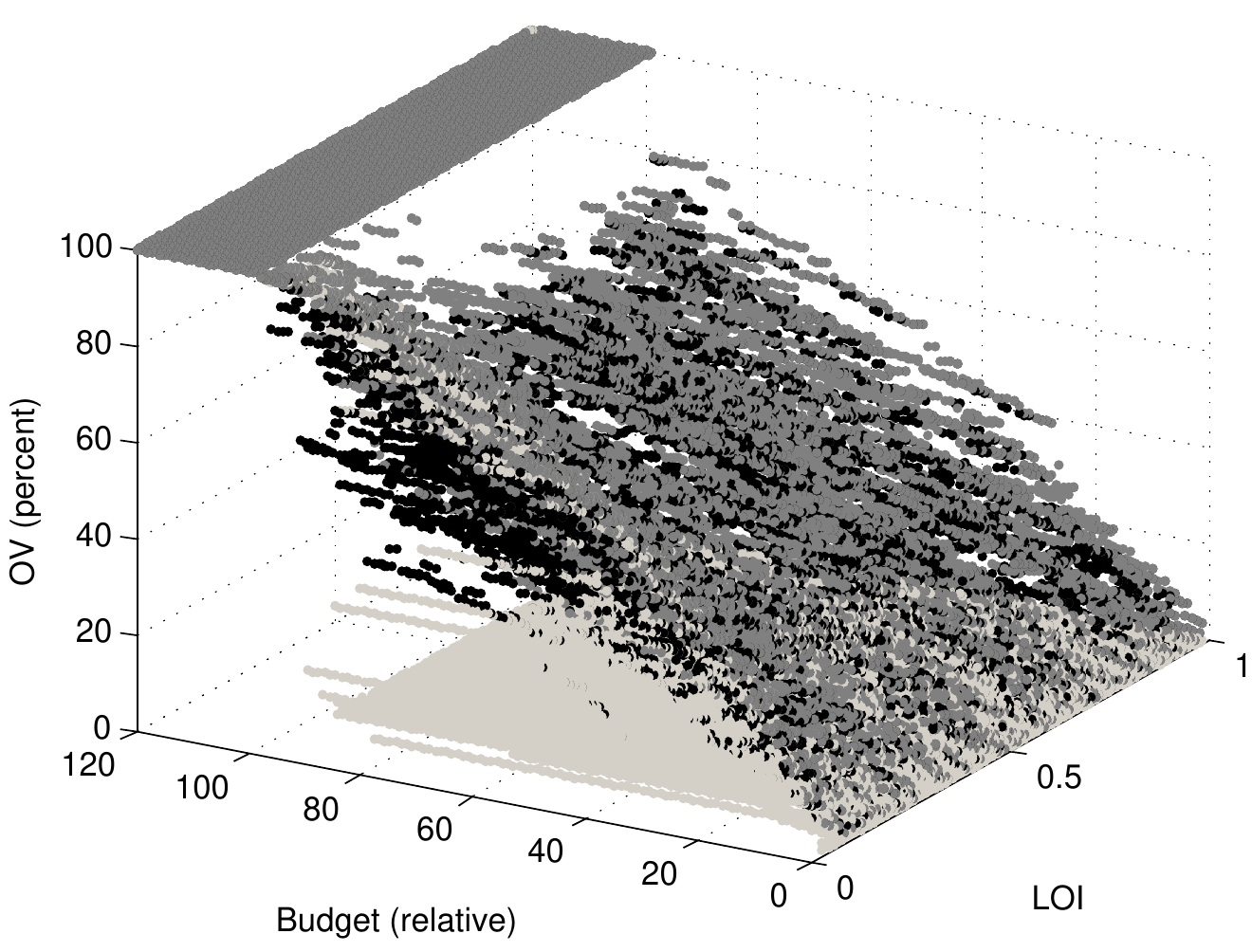}
		}%
		\subfigure[Overall value of PMR]{%
			\label{fig:second}
			\includegraphics[scale=0.5]{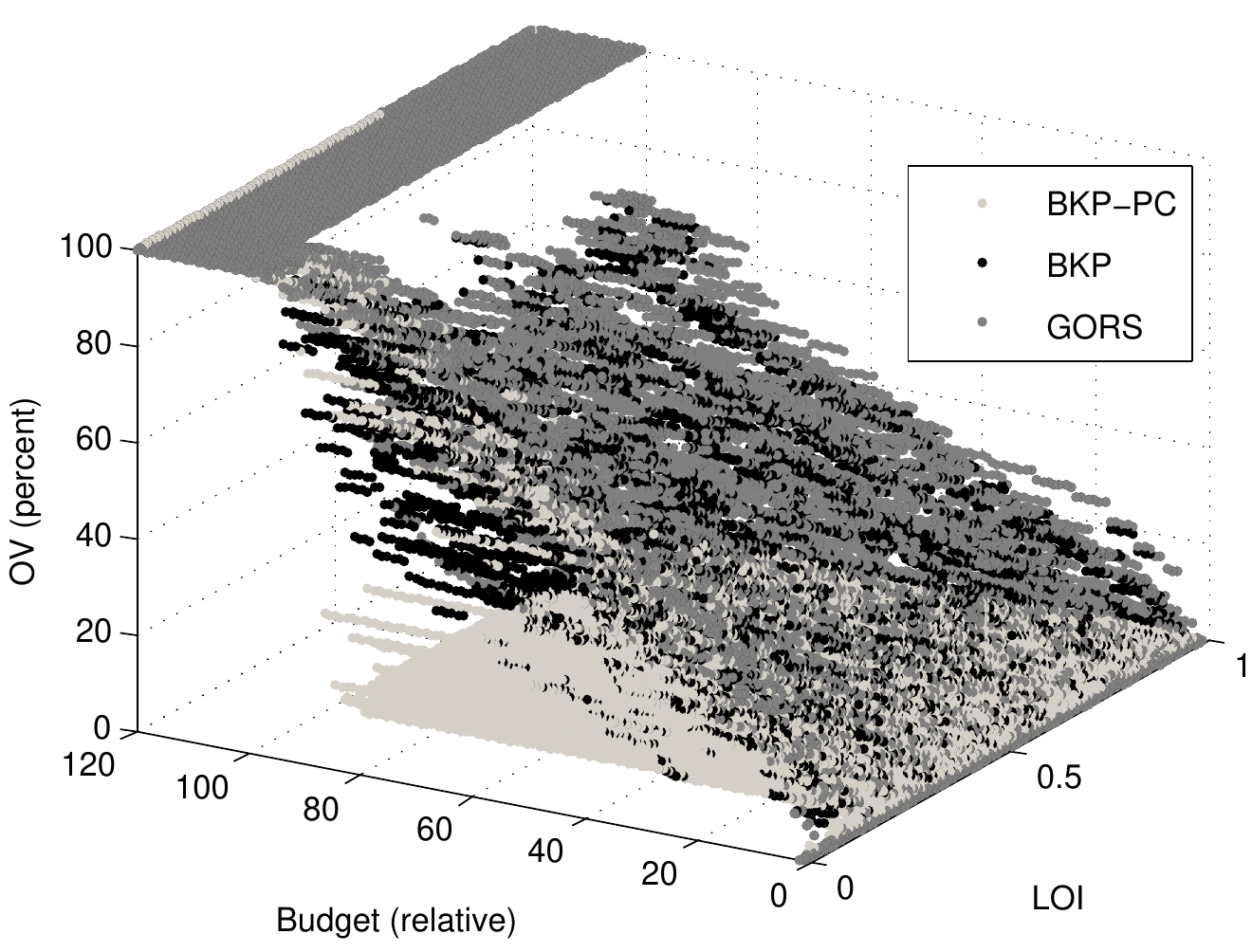}
		}
	\end{center}
	\caption{%
		Accumulated and overall values of RAN and PMR achieved from Simulations 
	}%
	\label{fig_result_all}
\end{figure*}

\begin{figure*}[!htb]
\begin{center}
\subfigure[LOI = $0.8$]{%
\label{fig:first}
\includegraphics[scale=0.5]{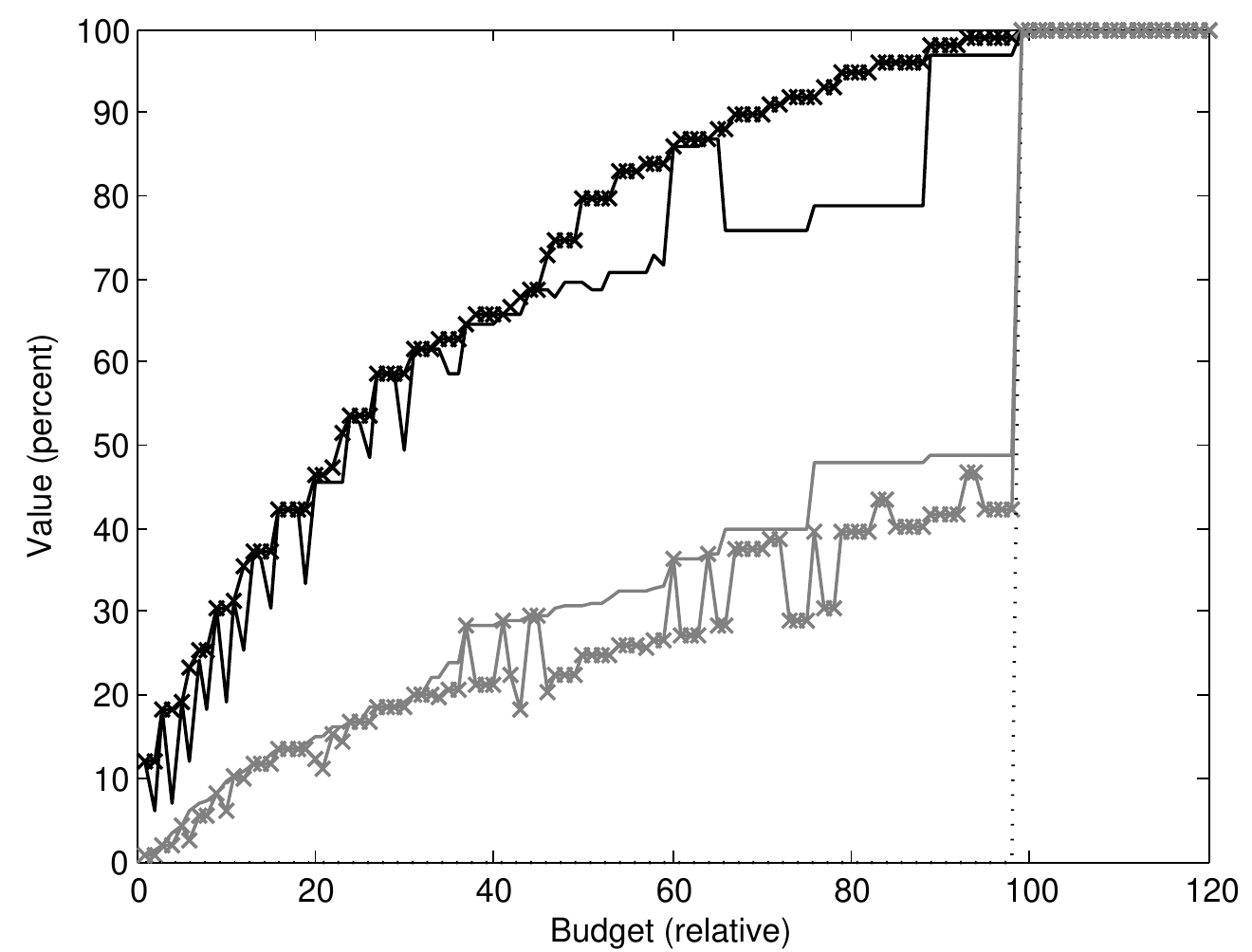}
}%
\subfigure[LOI = $0.4$]{%
\label{fig:second}
\includegraphics[scale=0.49]{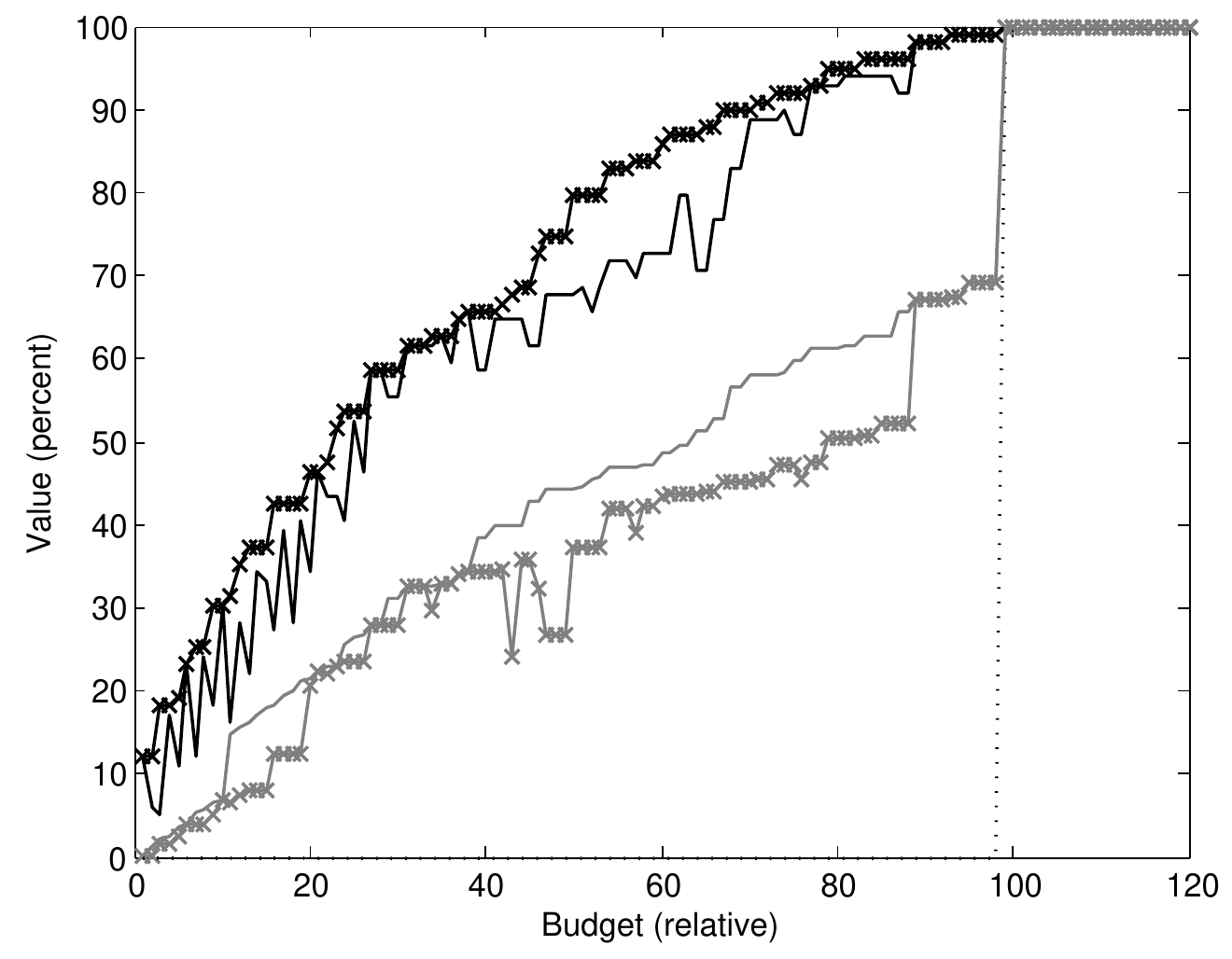}
}\\ 
\subfigure[LOI = $0.2$]{%
\label{fig:third}
\includegraphics[scale=0.5]{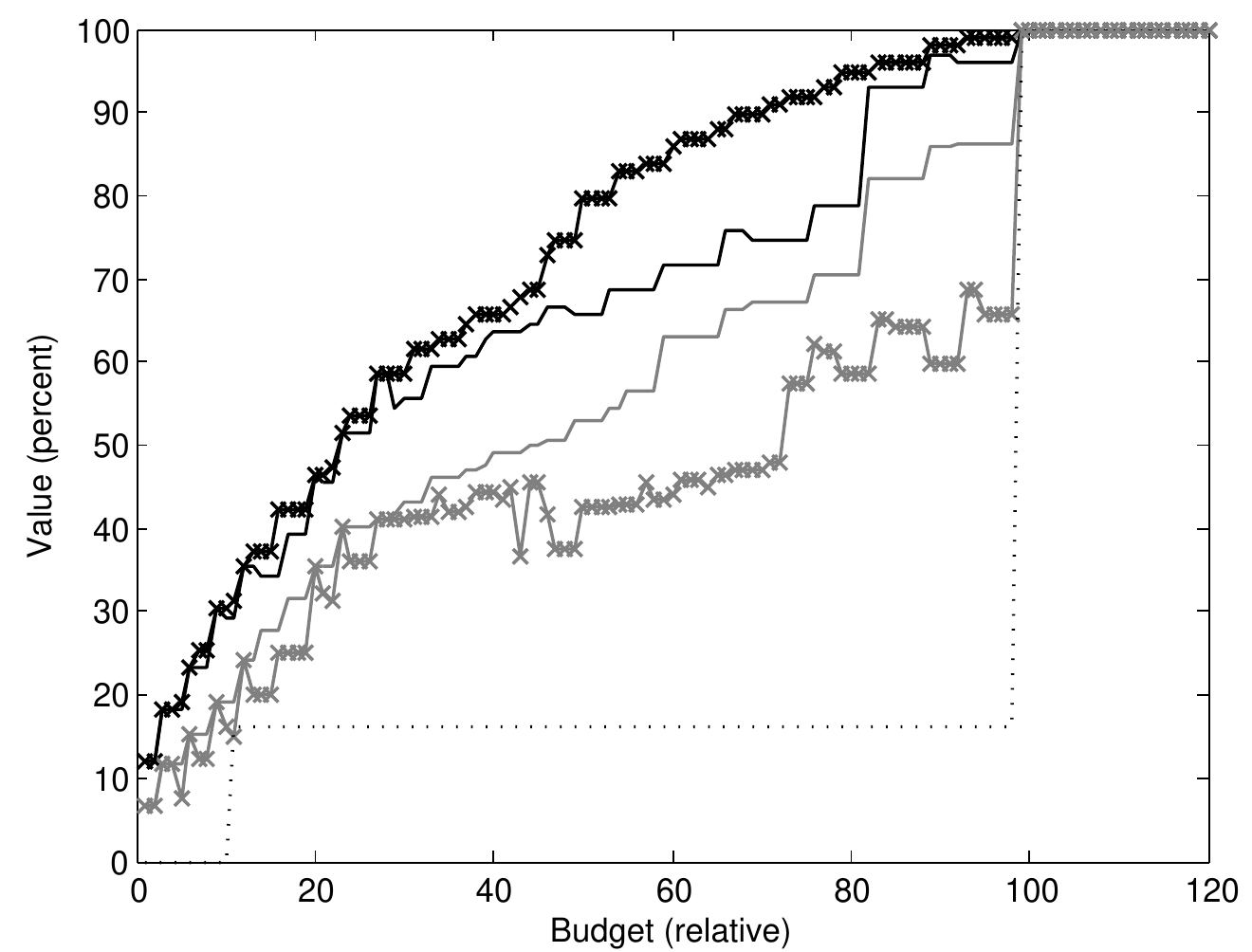}
}%
\subfigure[LOI = $0.05$]{%
\label{fig:fourth}
\includegraphics[scale=0.5]{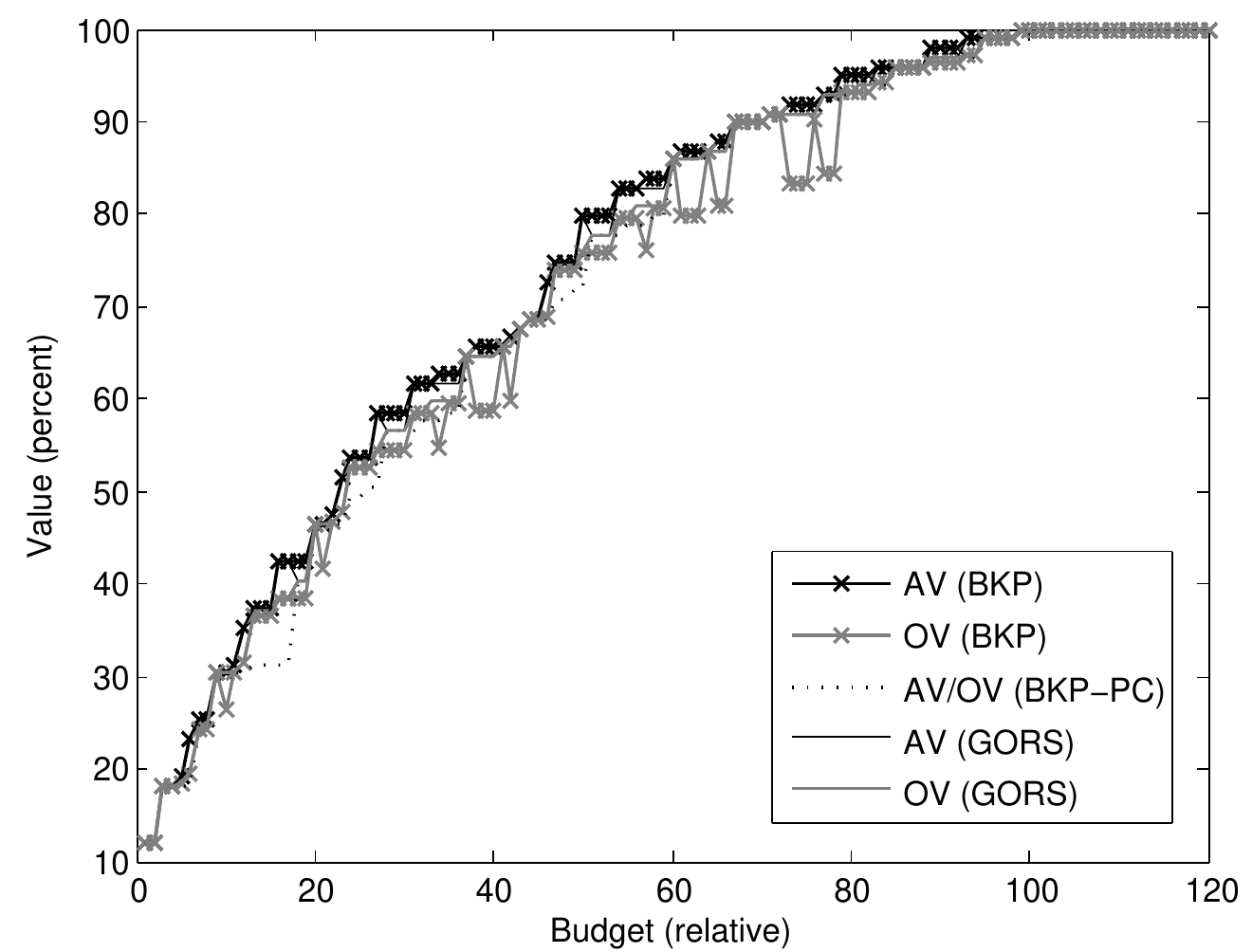}
}%
\end{center}
\caption{%
Sample simulation results for RAN requirements
}%
\label{fig_result_ran}
\end{figure*}
\begin{figure*}
\begin{center}
\subfigure[LOI = $0.8$]{%
\label{fig:first}
\includegraphics[scale=0.5]{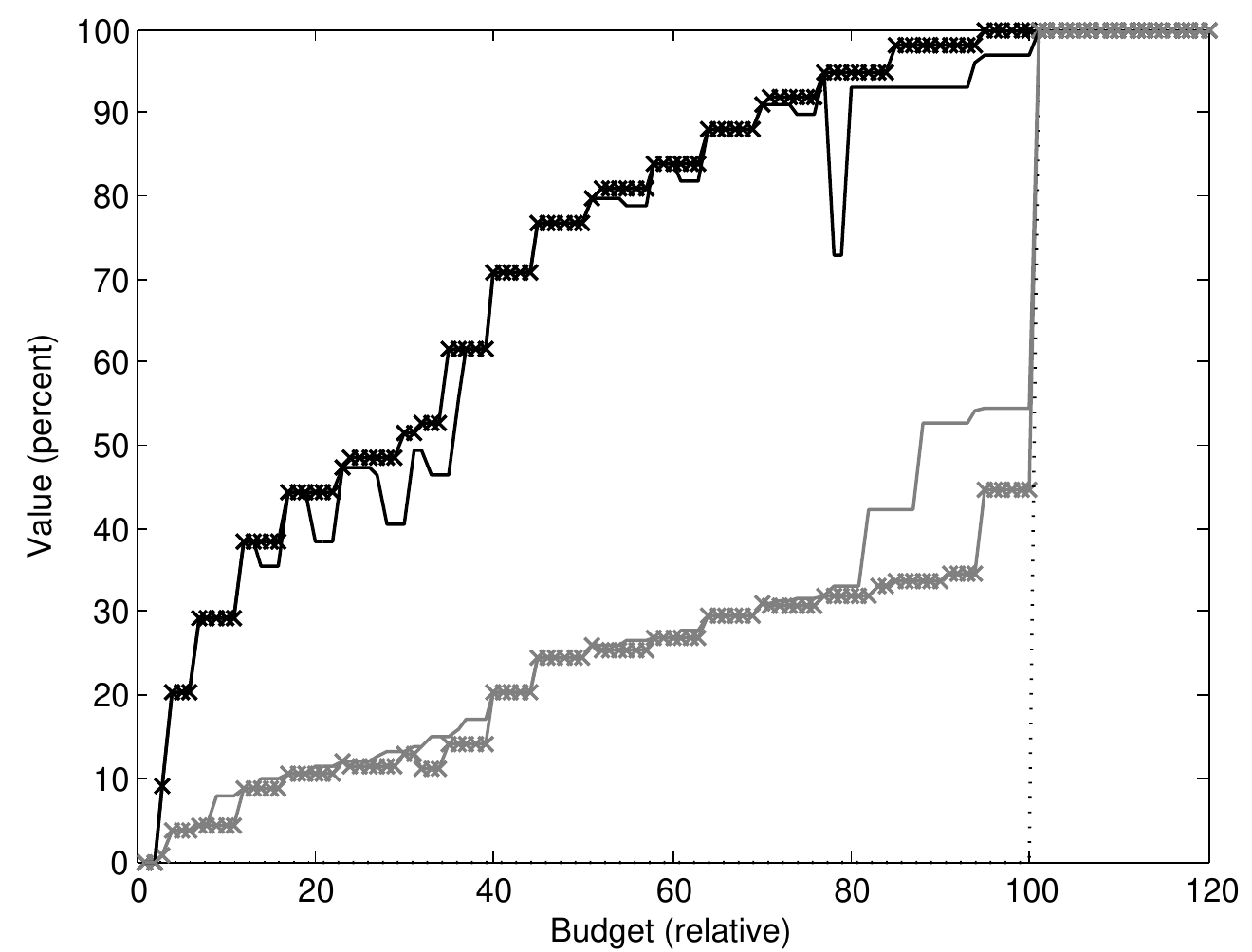}
}%
\subfigure[LOI = $0.4$]{%
\label{fig:second}
\includegraphics[scale=0.5]{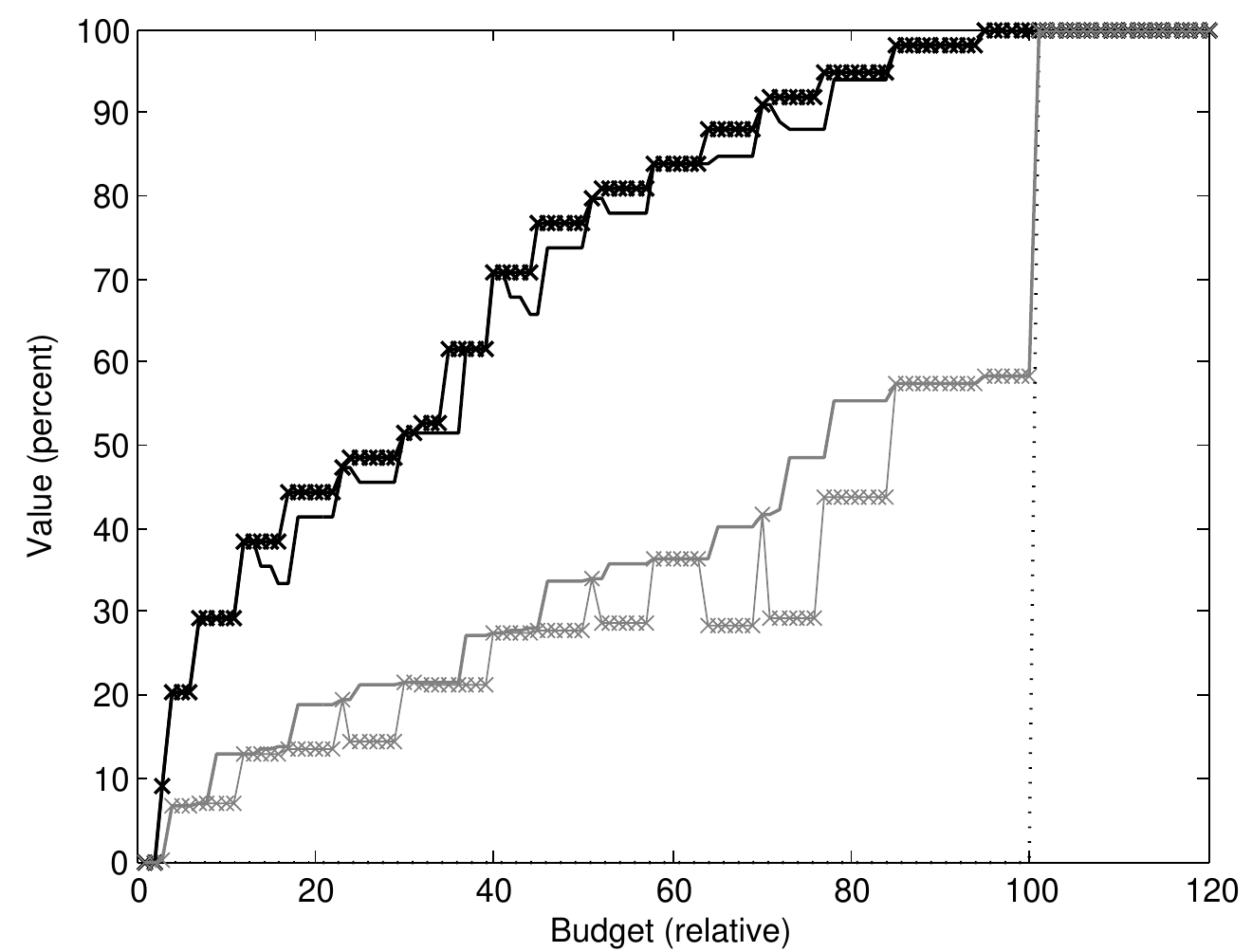}
}\\ 
\subfigure[LOI = $0.2$]{%
\label{fig:third}
\includegraphics[scale=0.49]{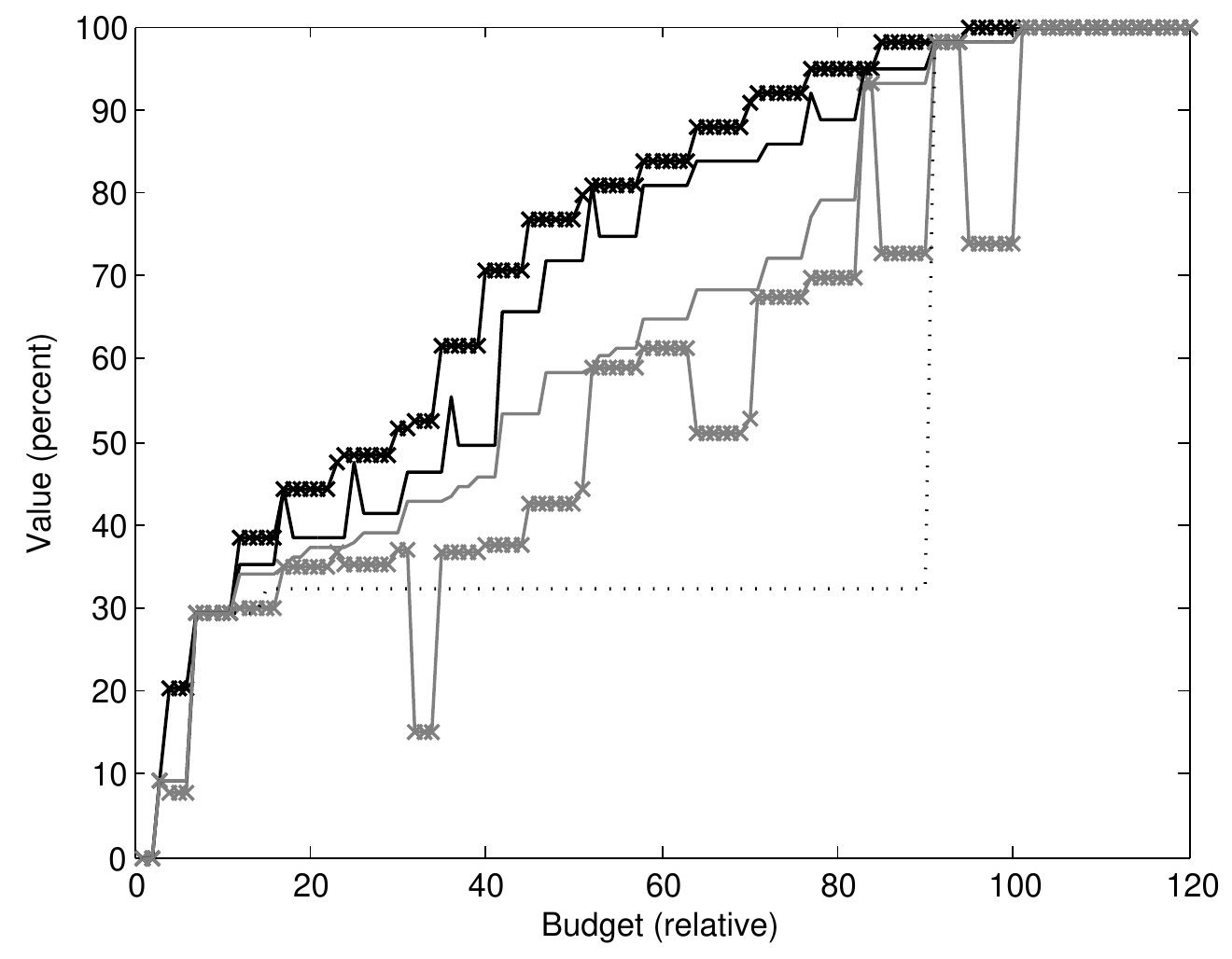}
}%
\subfigure[LOI = $0.05$]{%
\label{fig:fourth}
\includegraphics[scale=0.5]{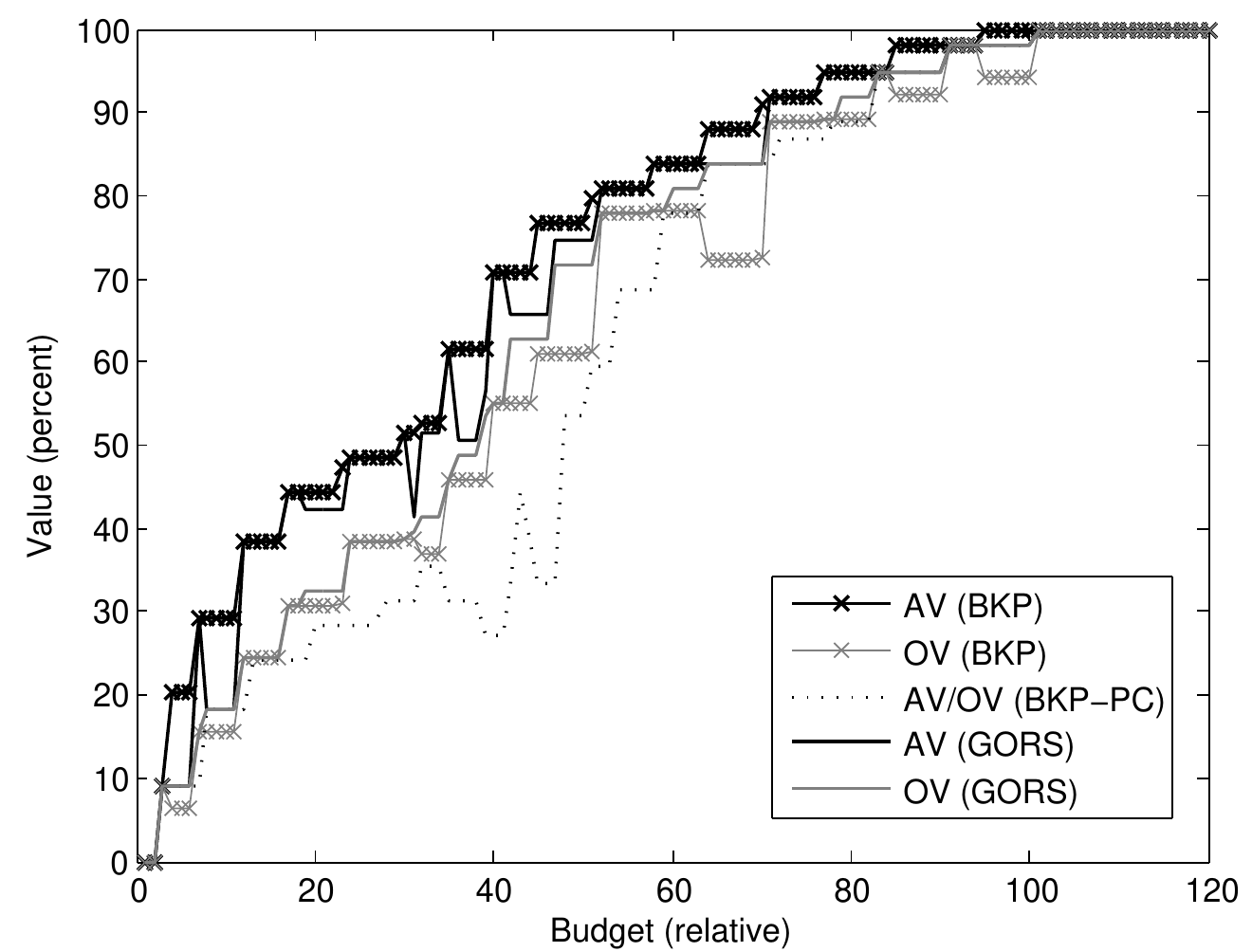}
}%
\end{center}
\caption{%
Sample simulation results for PMR requirements
}%
\label{fig_result_pmr}
\end{figure*}
The results of our simulations also showed (Figure \ref{fig_result_all}) that the efficiency of the BKP-PC model was severely impacted by the selection deficiency problem (SDP). In other words, the BKP-PC model generated the lowest $AV/OV$ unless in the presence of a sufficient budget ($Budget \to 120$) and/or a negligible level of interdependency ($LOI \to 0$). 

For $LOI > 0.25$ in both RAN and PMR requirement sets, almost no $AV/OV$ was achieved by the BKP-PC model unless budget was available for all of the requirements ($b=\sum_{i=1}^{14} c_i = 99$ for the RAN). It was moreover, observed (Figure \ref{fig_result_all}) that the GORS model mitigated the impact of the SDP through considering the strengths of value-related dependencies. 

The BKP model however was not subject to the SDP as it completely ignored dependencies among requirements.

We further observed (Figure~\ref{fig_result_all}) that all the of the selection models performed similar when budget was available for all of the requirements to be implemented ($Budget \geq 99$ for RAN and $Budget \geq$ 101 for PMR) or requirements were mutually independent.    

Figure~\ref{fig_result_ran} and Figure~\ref{fig_result_pmr} compare AV/OV achieved by the simulated selection models for various levels of interdependencies among requirements of the RAN and PMR respectively. A dependency level of $LOI=0.8$ implies that $80\%$ of the explicit value-related dependencies have non-zero strengths. The horizontal axis shows the available budget and the vertical axis shows the percentage of the achieved $AV/OV$. 

In almost every simulation, it was observed that for a given optimal set $O$, AV of $O$ was smaller or equal to the OV of $O$. This is due to the fact that the overall value of an optimal set considers the impacts of value-related dependencies among requirements whereas the accumulated value of an optimal set accumulates the estimated values of selected requirements without considering their value-related dependencies. 

It was further observed that the gap between overall value of an optimal set and its corresponding accumulated value ($|AV-OV|$) increased as the level of interdependency (LOI) grew. The reason is that increasing the LOI increases the chances that selected requirements explicitly depend on the excluded requirements which generally results in decreasing the of overall value of the optimal set. 

The BKP-PC model however always avoids choosing a requirement without its dependencies being selected. In other words, the BKP-PC model avoids dependencies from the optimal set to the excluded set and as such $AV = OV$ always holds for the BKP-PC models.

%% file: validation_case.tex
\subsection{Case Study}
\label{sec_case}
To demonstrate practicality of the GORS model, we performed selection for $23$ requirements of a messaging software referred to as the \textit{Precious Messaging System (\gls{PMS})}. We employed $5$ stakeholders to estimate~\cite{greer_software_2004} the costs and values of requirements of the PMS. Each requirement $r_i$ was assigned an estimated cost of $c_i \in [1,20]$ and an estimated value of $v_i \in [1,20]$ by different stakeholders. 


Stakeholders then, performed pairwise comparisons among requirements~\cite{carlshamre_industrial_2001} to identify explicit value-related dependencies and estimate the strengths of those dependencies. A dependency $(r_i,r_j)$ was assigned an strength of $\rho(r_i,r_j) \in [0,1]$ where $\rho(r_i,r_j)=0$ and $\rho(r_i,r_j)=1$ denoted no dependency and a full dependency from $r_i$ to $r_j$ respectively.  

The median of $5$ estimated costs/values for each requirement $r_i$ then was computed to account for different opinions of stakeholders. In a similar way, for each explicit value-related dependency $(r_i,r_j)$ the median of the $5$ estimated strengths of that dependency was computed to specify the strength of $(r_i,r_j)$. Median was taken as the measure of central tendency as it is less affected by (potentially) extreme opinions of stakeholders compared to the arithmetic mean. 

Table~\ref{table_cs_dependencies} lists the estimated costs and values of the requirement of the PMS as well as the strengths of explicit value-related dependencies among those requirements. The \textit{Dependency Vector} of a requirement $r_i$ in Table \ref{table_cs_dependencies} denotes the strengths of explicit value-related dependencies from $r_i$ to other requirements of the PMS. Based on Table \ref{table_cs_dependencies} and (\ref{Eq_loi}), level of interdependency is calculated for the requirements of the PMS as follows. $LOI(PMS)=\frac{113}{\perm{23}{2}} \approxeq 0.22$.

\begin{table*}[!htb]
	\caption{Estimated values, costs, and strengths of explicit value-related dependencies (based on stakeholder's estimations).}
	\label{table_cs_dependencies}
	\centering
	\input{table_cs_dependencies}
\end{table*}

\begin{table*}[!htb]
		\caption{Solution vectors and their corresponding overall value (OV) provided by the experimented selection models in the presence of various budget constraints. A selection variable $x_i$ denotes whether requirement $r_i$ is selected ($x_i=1$) or otherwise ($x_i=0$)}
		\label{table_cs_solutions}
		\centering
		\input{table_cs_solutions}
\end{table*}

Based on the estimations provided by the stakeholders, FRIG of the PMS was constructed (Figure~\ref{fig_cs}) and selections were performed using the GORS model as well as the BKP and BKP-PC models. Requirement selections were performed for various ranges of budgets ($Budget \in \{1,...,260\}$) to examine the performance of the selection models. 

Figure \ref{fig_cs_result} summarizes the results of our experiments by comparing the accumulated values (AV) and/or overall values (OV) achieved by the selection models. The horizontal axis shows the available budget ($Budget=\{1,...,260\}$) and the vertical axis shows the percentages of $AV/OV$. Table~\ref{table_cs_solutions} lists some of the optimal sets provided by the employed selection models in the presence of various budget constraints. 

Consistent with the simulations, the results of our case study demonstrated (Figure \ref{fig_cs_result} and Table \ref{table_cs_solutions}) that the BKP model always maximized the accumulated value of the selected requirements (optimal set) while the GORS model maximized the overall value of selected requirements. Moreover, maximizing accumulated value and overall value of an optimal set demonstrated to be conflicting objectives.

\begin{figure}[!htb]
	\centering
	\centerline{\hspace{-0.0em}\includegraphics[scale=0.5]{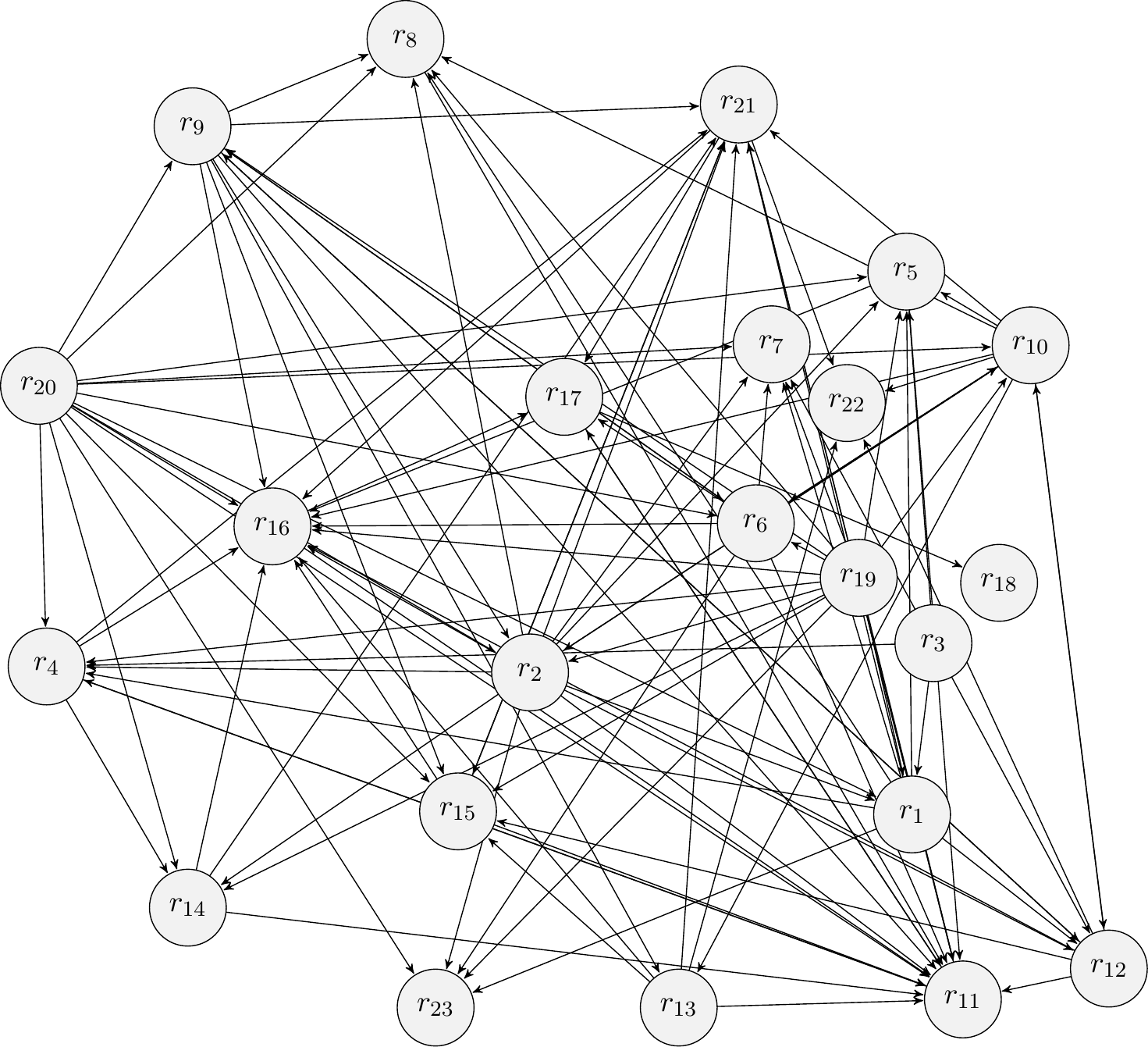}}
	\caption{FRIG of the PMS (Strengths of dependencies are not represented for the sake of readability)}
	\label{fig_cs}
\end{figure}

Furthermore, the results of our experiments showed (Figure \ref{fig_cs_result} and Table \ref{table_cs_solutions}) that the GORS model mitigated the adverse impact of the selection deficiency problem (SDP) through considering the strengths of value-related requirement dependencies while the efficiency of the BKP-PC model was negatively impacted by the SDP. For instance, we observed (Table~\ref{table_cs_solutions}) that for $Budget =81$, overall value of the optimal set provided by the GORS model was almost twice as higher as the overall value provided by the BKP-PC model. The BKP model on the contrary, was not vulnerable to the SDP as it totally ignores dependencies.


\begin{figure}[h!]
	\centering
	\centerline{\hspace{-0.5em}\includegraphics[scale=0.6]{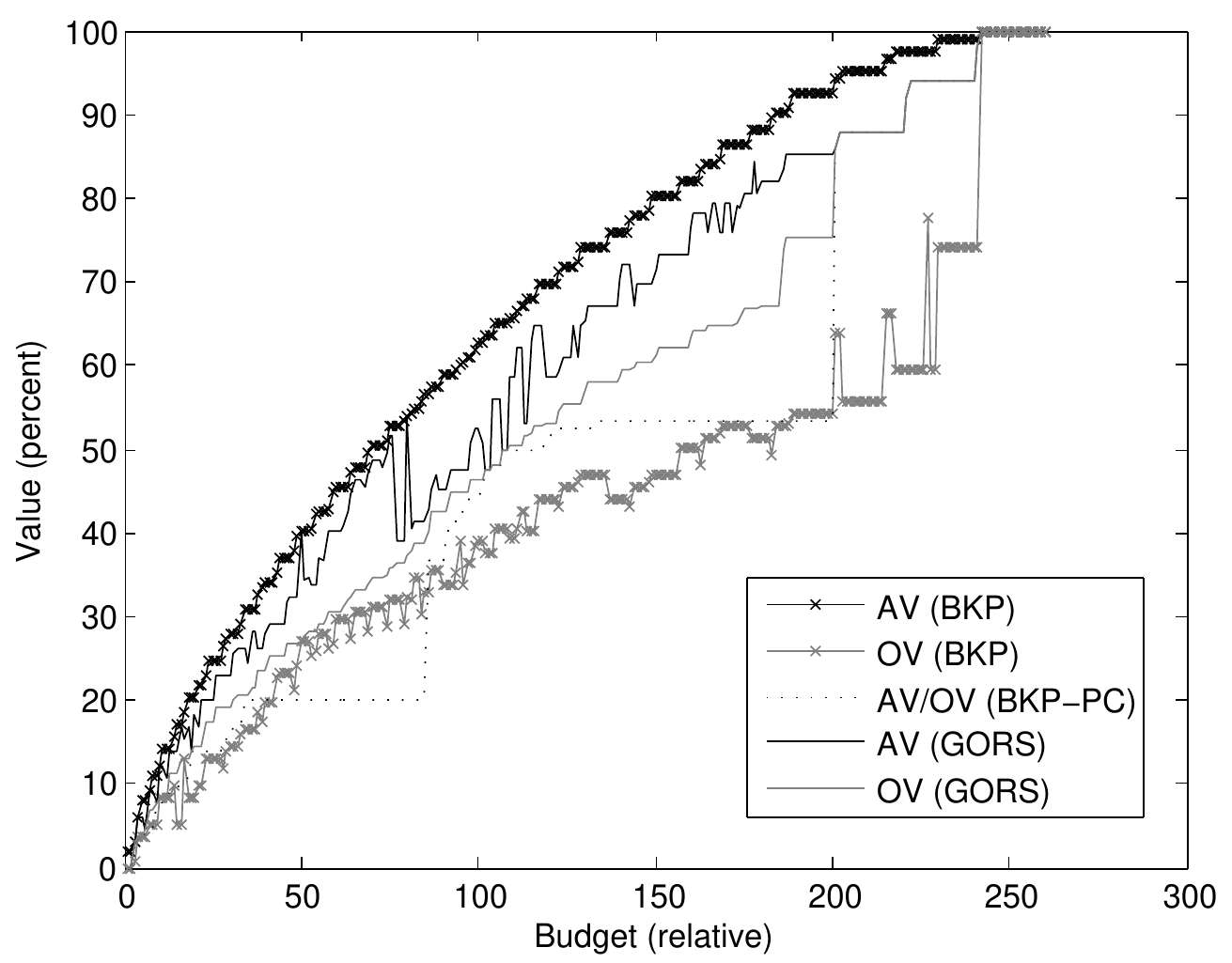}}
	\caption{Selection results for the PMS (LOI $\approxeq 22\%$)}
	\label{fig_cs_result}
\end{figure}

%% file: table_cs_dependencies.tex
\resizebox {0.9\textwidth }{!}{
	\begin{tabular}{llll}
    \toprule[1.5pt]
	\textbf{ID} &
	\textbf{Value} &
	\textbf{Cost} &
	\textbf{Dependency Vector $\{r_1,...,r_{23}\}$}
	\\
	\midrule
    
    \rowcolor{white}
    $$ &
    $$ &
    $$ &
    $$
    \\
	$r_1$ &
	$20$ &
	$10$ &
	$\{0.0,0.0,0.0,0.5,0.3,0.0,0.6,0.4,0.0,0.0,0.0,0.7,0.0,0.0,0.0,0.0,0.0,0.0,0.0,0.0,0.1,0.0,0.3\}$
	\\
	$r_2$ &
	$20$ &
	$7$ &
	$\{1.0,0.0,0.0,0.6,0.6,0.0,0.6,0.6,0.0,0.3,0.3,0.7,0.0,0.3,0.0,0.7,0.0,0.0,0.0,0.0,0.2,0.0,0.8\}$
	\\
	$r_3$ &
	$6$ &
	$1$ &
	$\{1.0,0.0,0.0,0.5,0.3,0.0,0.6,0.0,0.0,0.0,0.0,0.7,0.0,0.0,0.0,0.0,0.0,0.0,0.0,0.0,0.0,0.0,0.0\}$
	\\
	$r_4$ &
	$17$ &
	$10$ &
	$\{0.0,0.0,0.0,0.0,0.0,0.0,0.0,0.0,0.0,0.0,0.9,0.0,0.0,0.4,0.0,0.7,0.0,0.0,0.0,0.0,0.2,0.0,0.0\}$
	\\
	$r_5$ &
	$3$ &
	$12$ &
	$\{0.0,0.0,0.0,0.0,0.0,0.0,0.0,0.0,0.0,0.0,0.5,0.0,0.0,0.0,0.0,0.2,0.0,0.0,0.0,0.0,0.0,0.0,0.0\}$
	\\
	$r_6$ &
	$20$ &
	$20$ &
	$\{0.0,0.0,0.0,0.0,0.0,0.0,0.6,0.0,0.3,0.3,0.4,0.0,0.0,0.0,0.0,0.7,0.4,0.0,0.0,0.0,0.0,0.0,0.8\}$
	\\
	$r_7$ &
	$15$ &
	$6$ &
	$\{0.0,0.0,0.0,0.0,0.0,0.0,0.0,0.0,0.0,0.0,0.0,0.0,0.0,0.0,0.0,0.0,0.0,0.0,0.0,0.0,0.0,0.0,0.0\}$
	\\
	$r_8$ &
	$8$ &
	$14$ &
	$\{0.0,0.0,0.0,0.0,0.0,0.0,0.0,0.0,0.0,0.0,0.4,0.0,0.0,0.0,0.0,0.0,0.0,0.0,0.0,0.0,0.0,0.0,0.0\}$
	\\
	$r_9$ &
	$20$ &
	$15$ &
	$\{0.0,0.7,0.0,0.0,0.0,0.7,0.0,0.3,0.0,0.0,0.8,0.2,0.4,0.0,0.2,0.7,0.0,0.0,0.0,0.0,0.2,0.0,0.0\}$
	\\
	$r_{10}$ &
	$16$ &
	$10$ &
	$\{0.0,0.7,0.0,0.0,0.3,0.7,0.0,0.3,0.0,0.0,0.0,0.3,0.4,0.0,0.0,0.2,0.0,0.0,0.0,0.0,0.2,0.6,0.0\}$
	\\
	$r_{11}$ &
	$20$ &
	$4$ &
	$\{0.0,0.0,0.0,0.4,0.0,0.0,0.0,0.0,0.0,0.0,0.0,0.0,0.0,0.0,0.0,0.0,0.4,0.0,0.0,0.0,0.2,0.0,0.0\}$
	\\
	$r_{12}$ &
	$10$ &
	$6$ &
	$\{0.0,0.0,0.0,0.0,0.0,0.0,0.0,0.0,0.5,0.5,0.8,0.0,0.0,0.0,0.1,0.4,0.0,0.0,0.0,0.0,0.0,0.4,0.0\}$
	\\
	$r_{13}$ &
	$8$ &
	$5$ &
	$\{0.0,0.0,0.0,0.0,0.0,0.0,0.0,0.0,0.0,0.0,0.8,0.0,0.0,0.0,0.1,0.4,0.0,0.0,0.0,0.0,0.1,0.6,0.0\}$
	\\
	$r_{14}$ &
	$5$ &
	$12$ &
	$\{0.0,0.0,0.0,0.0,0.0,0.0,0.0,0.0,0.0,0.0,0.8,0.0,0.0,0.0,0.0,0.8,0.0,0.0,0.0,0.0,0.8,0.0,0.0\}$
	\\
	$r_{15}$ &
	$8$ &
	$15$ &
	$\{0.0,0.0,0.0,0.0,0.0,0.0,0.0,0.0,0.0,0.0,0.8,0.0,0.0,0.0,0.0,0.8,0.0,0.0,0.0,0.0,1.0,0.0,0.0\}$
	\\
	$r_{16}$ &
	$10$ &
	$3$ &
	$\{0.0,0.0,0.0,0.0,0.0,0.0,0.0,0.0,0.0,0.0,0.8,0.0,0.0,0.0,0.0,0.0,0.5,0.0,0.0,0.0,0.0,0.0,0.0\}$
	\\
	$r_{17}$ &
	$15$ &
	$12$ &
	$\{0.0,0.0,0.0,0.0,0.0,0.0,0.0,0.0,0.0,0.0,0.8,0.0,0.0,0.0,0.0,0.0,0.0,0.5,0.0,0.0,0.0,0.0,0.0\}$
	\\
	$r_{18}$ &
	$10$ &
	$3$ &
	$\{0.0,0.0,0.0,0.0,0.0,0.0,0.0,0.0,0.0,0.0,0.0,0.0,0.0,0.0,0.0,0.0,0.0,0.0,0.0,0.0,0.0,0.0,0.0\}$
	\\
	$r_{19}$ &
	$20$ &
	$20$ &
	$\{1.0,0.3,0.0,0.7,0.5,1.0,0.6,0.5,1.0,0.6,0.4,0.0,0.0,0.1,0.8,0.8,0.0,0.0,0.0,0.0,0.0,0.0,0.8\}$
	\\
	$r_{20}$ &
	$20$ &
	$20$ &
	$\{1.0,0.3,0.0,0.7,0.5,1.0,0.6,0.5,1.0,0.6,0.4,0.0,0.0,0.1,0.8,0.8,0.0,0.0,0.0,0.0,0.0,0.0,0.8\}$
	\\
	$r_{21}$ &
	$15$ &
	$12$ &
	$\{0.0,0.0,0.0,0.0,0.0,0.0,0.0,0.0,0.0,0.0,0.8,0.0,0.0,0.0,0.1,0.8,0.2,0.0,0.0,0.0,0.0,0.3,0.0\}$
	\\
	$r_{22}$ &
	$20$ &
	$15$ &
	$\{0.0,0.0,0.0,0.0,0.0,0.0,0.0,0.0,0.0,0.0,0.0,0.0,0.0,0.0,0.0,0.0,0.0,0.0,0.0,0.0,0.0,0.0,0.0\}$
	\\
	$r_{23}$ &
	$20$ &
	$10$ &
	$\{0.0,0.0,0.0,0.0,0.0,0.0,0.0,0.0,0.0,0.0,0.0,0.0,0.0,0.0,0.0,0.0,0.0,0.0,0.0,0.0,0.0,0.0,0.0\}$
    \\[2ex]
    \bottomrule[1.5pt]
\end{tabular}}%

%% file: table_cs_solutions.tex
\resizebox {0.88\textwidth }{!}{
	\begin{tabular}{llll}
		\toprule[1.5pt]
		\textbf{Budget} &
		\textbf{Selection Model} &
		\textbf{Overall Value (percent)} &
		\boldmath{}\textbf{Solution Vector $\{x_1,...,x_{23}\}$}\unboldmath{}
		\\ \midrule[1.5pt]
		\multirow{3}[1]{*}{16} &
		BKP &
		5.21 &
		$\{0,1,1,0,0,0,0,0,0,0,1,0,0,0,0,1,0,0,0,0,0,0,0\}$
		\\
		&
		BKP-PC &
		10.74 &
		$\{0,0,0,0,0,0,1,0,0,0,0,0,0,0,0,0,0,0,0,0,0,0,1\}$
		\\
		&
		GORS &
		12.88 &
		$\{0,0,0,0,0,0,1,0,0,0,1,0,0,0,0,1,0,1,0,0,0,0,0\}$
		\bigstrut[b]\\
		\hline
		\multirow{3}[2]{*}{46} &
		BKP &
		23.25 &
		$\{1,1,1,0,0,0,1,0,0,0,1,0,0,0,0,1,0,1,0,0,0,0,1\}$
		\bigstrut[t]\\
		&
		BKP-PC &
		19.94 &
		$\{0,0,0,0,0,0,1,0,0,0,0,0,0,0,0,0,0,1,0,0,0,1,1\}$
		\\
		&
		GORS &
		26.63 &
		$\{0,0,0,0,0,0,1,0,0,0,1,0,1,0,0,1,0,1,0,0,0,1,1\}$
		\bigstrut[b]\\
		\hline
		\multirow{3}[2]{*}{71} &
		BKP &
		31.07 &
		$\{1,1,1,1,0,0,1,0,0,1,1,1,0,0,0,1,0,1,0,0,0,0,1\}$
		\bigstrut[t]\\
		&
		BKP-PC &
		19.94 &
		$\{0,0,0,0,0,0,1,0,0,0,0,0,0,0,0,0,0,1,0,0,0,1,1\}$
		\\
		&
		GORS &
		34.60 &
		$\{1,1,1,0,0,0,1,0,0,0,1,1,1,0,0,1,0,1,0,0,0,1,1\}$
		\bigstrut[b]\\
		\hline
		\multirow{3}[2]{*}{76} &
		BKP &
		32.06 &
		$\{1,1,1,1,0,0,1,0,0,1,1,1,1,0,0,1,0,1,0,0,0,0,1\}$
		\bigstrut[t]\\
		&
		BKP-PC &
		19.94 &
		$\{0,0,0,0,0,0,1,0,0,0,0,0,0,0,0,0,0,1,0,0,0,1,1\}$
		\\
		&
		GORS &
		35.74 &
		$\{1,1,1,1,0,0,1,0,0,0,1,1,0,0,0,1,0,1,0,0,0,1,1\}$
		\bigstrut[b]\\
		\hline
		\multirow{3}[2]{*}{81} &
		BKP &
		31.90 &
		$\{1,1,1,1,0,0,1,0,0,1,1,0,1,0,0,1,1,1,0,0,0,0,1\}$
		\bigstrut[t]\\
		&
		BKP-PC &
		19.94 &
		$\{0,0,0,0,0,0,1,0,0,0,0,0,0,0,0,0,0,1,0,0,0,1,1\}$
		\\
		&
		GORS &
		37.98 &
		$\{0,0,0,1,0,0,0,0,0,0,1,0,0,1,0,1,1,1,0,0,1,1,1\}$
		\bigstrut[b]\\
		\hline
		\multirow{3}[2]{*}{141} &
		BKP &
		44.11 &
		$\{1,1,1,1,0,0,1,0,1,1,1,1,1,0,0,1,1,1,0,1,0,1,1\}$
		\bigstrut[t]\\
		&
		BKP-PC &
		53.37 &
		$\{0,0,0,1,1,0,1,1,0,0,1,0,1,1,1,1,1,1,0,0,1,1,1\}$
		\\
		&
		GORS &
		59.45 &
		$\{1,1,1,1,0,0,1,0,0,1,1,1,1,1,1,1,1,1,0,0,1,1,1\}$
		\bigstrut[b]\\
		\hline
		\multirow{3}[2]{*}{146} &
		BKP &
		45.40 &
		$\{1,1,1,1,0,0,1,0,1,1,1,1,0,0,0,1,1,1,0,1,1,1,1\}$
		\bigstrut[t]\\
		&
		BKP-PC &
		53.37 &
		$\{0,0,0,1,1,0,1,1,0,0,1,0,1,1,1,1,1,1,0,0,1,1,1\}$
		\\
		&
		GORS &
		60.43 &
		$\{1,1,1,1,0,0,1,1,0,0,1,1,1,1,1,1,1,1,0,0,1,1,1\}$
		\bigstrut[b]\\
		\hline
		\multirow{3}[2]{*}{151} &
		BKP &
		46.87 &
		$\{1,1,1,1,0,0,1,0,1,1,1,1,1,0,0,1,1,1,0,1,1,1,1\}$
		\bigstrut[t]\\
		&
		BKP-PC &
		53.37 &
		$\{0,0,0,1,1,0,1,1,0,0,1,0,1,1,1,1,1,1,0,0,1,1,1\}$
		\\
		&
		GORS &
		62.27 &
		$\{1,1,1,1,0,1,1,0,0,0,1,1,1,1,1,1,1,1,0,0,1,1,1\}$
		\bigstrut[b]\\
		\hline
		\multirow{3}[2]{*}{156} &
		BKP &
		46.87 &
		$\{1,1,1,1,0,0,1,0,1,1,1,1,1,0,0,1,1,1,0,1,1,1,1\}$
		\bigstrut[t]\\
		&
		BKP-PC &
		53.37 &
		$\{0,0,0,1,1,0,1,1,0,0,1,0,1,1,1,1,1,1,0,0,1,1,1\}$
		\\
		&
		GORS &
		62.27 &
		$\{1,1,1,1,0,1,1,0,0,0,1,1,1,1,1,1,1,1,0,0,1,1,1\}$
		\bigstrut[b]\\
		\hline
		\multirow{3}[2]{*}{161} &
		BKP &
		50.12 &
		$\{1,1,1,1,0,1,1,0,1,1,1,1,1,0,0,1,1,1,0,1,0,1,1\}$
		\bigstrut[t]\\
		&
		BKP-PC &
		53.37 &
		$\{0,0,0,1,1,0,1,1,0,0,1,0,1,1,1,1,1,1,0,0,1,1,1\}$
		\\
		&
		GORS &
		64.23 &
		$\{1,1,1,1,0,1,1,0,0,1,1,1,1,1,1,1,1,1,0,0,1,1,1\}$
		\bigstrut[b]\\
		\hline
		\multirow{3}[2]{*}{166} &
		BKP &
		51.41 &
		$\{1,1,1,1,0,1,1,0,1,1,1,1,0,0,0,1,1,1,0,1,1,1,1\}$
		\bigstrut[t]\\
		&
		BKP-PC &
		53.37 &
		$\{0,0,0,1,1,0,1,1,0,0,1,0,1,1,1,1,1,1,0,0,1,1,1\}$
		\\
		&
		GORS &
		64.72 &
		$\{1,1,1,1,0,1,1,0,1,0,1,1,1,1,1,1,1,1,0,0,1,1,1\}$
		\bigstrut[b]\\
		\hline
		\multirow{3}[2]{*}{171} &
		BKP &
		52.88 &
		$\{1,1,1,1,0,1,1,0,1,1,1,1,1,0,0,1,1,1,0,1,1,1,1\}$
		\bigstrut[t]\\
		&
		BKP-PC &
		53.37 &
		$\{0,0,0,1,1,0,1,1,0,0,1,0,1,1,1,1,1,1,0,0,1,1,1\}$
		\\
		&
		GORS &
		64.72 &
		$\{1,1,1,1,0,1,1,0,1,0,1,1,1,1,1,1,1,1,0,0,1,1,1\}$
		\bigstrut[b]\\
		\hline
		\multirow{3}[2]{*}{176} &
		BKP &
		52.88 &
		$\{1,1,1,1,0,1,1,0,1,1,1,1,1,0,0,1,1,1,0,1,1,1,1\}$
		\bigstrut[t]\\
		&
		BKP-PC &
		53.37 &
		$\{0,0,0,1,1,0,1,1,0,0,1,0,1,1,1,1,1,1,0,0,1,1,1\}$
		\\
		&
		GORS &
		66.69 &
		$\{1,1,1,1,0,1,1,1,0,1,1,1,1,1,1,1,1,1,0,0,1,1,1\}$
		\bigstrut[b]\\
		\hline
		\multirow{3}[2]{*}{181} &
		BKP &
		51.35 &
		$\{1,1,1,1,0,1,1,0,1,1,1,1,1,0,0,1,1,1,1,1,0,1,1\}$
		\bigstrut[t]\\
		&
		BKP-PC &
		53.37 &
		$\{0,0,0,1,1,0,1,1,0,0,1,0,1,1,1,1,1,1,0,0,1,1,1\}$
		\\
		&
		GORS &
		67.18 &
		$\{1,1,1,1,0,1,1,1,1,0,1,1,1,1,1,1,1,1,0,0,1,1,1\}$
		\bigstrut[b]\\
		\hline
		\multirow{3}[2]{*}{186} &
		BKP &
		52.64 &
		$\{1,1,1,1,0,1,1,0,1,1,1,1,0,0,0,1,1,1,1,1,1,1,1\}$
		\bigstrut[t]\\
		&
		BKP-PC &
		53.37 &
		$\{0,0,0,1,1,0,1,1,0,0,1,0,1,1,1,1,1,1,0,0,1,1,1\}$
		\\
		&
		GORS &
		73.83 &
		$\{1,1,0,1,1,1,1,1,1,1,1,1,1,1,0,1,1,1,0,0,1,1,1\}$
		\bigstrut[b]\\
		\hline
		\multirow{3}[2]{*}{191} &
		BKP &
		54.11 &
		$\{1,1,1,1,0,1,1,0,1,1,1,1,1,0,0,1,1,1,1,1,1,1,1\}$
		\bigstrut[t]\\
		&
		BKP-PC &
		53.37 &
		$\{0,0,0,1,1,0,1,1,0,0,1,0,1,1,1,1,1,1,0,0,1,1,1\}$
		\\
		&
		GORS &
		75.31 &
		$\{1,1,1,1,1,1,1,1,1,1,1,1,1,1,0,1,1,1,0,0,1,1,1\}$
		\bigstrut[b]\\
		\hline
		\multirow{3}[2]{*}{196} &
		BKP &
		54.11 &
		$\{1,1,1,1,0,1,1,0,1,1,1,1,1,0,0,1,1,1,1,1,1,1,1\}$
		\bigstrut[t]\\
		&
		BKP-PC &
		53.37 &
		$\{0,0,0,1,1,0,1,1,0,0,1,0,1,1,1,1,1,1,0,0,1,1,1\}$
		\\
		&
		GORS &
		75.31 &
		$\{1,1,1,1,1,1,1,1,1,1,1,1,1,1,0,1,1,1,0,0,1,1,1\}$
		\bigstrut[b]\\
		\hline
		\multirow{3}[1]{*}{246} &
		BKP &
		100.00 &
		$\{1,1,1,1,1,1,1,1,1,1,1,1,1,1,1,1,1,1,1,1,1,1,1\}$
		\bigstrut[t]\\
		&
		BKP-PC &
		100.00 &
		$\{1,1,1,1,1,1,1,1,1,1,1,1,1,1,1,1,1,1,1,1,1,1,1\}$
		\\
		&
		GORS &
		100.00 &
		$\{1,1,1,1,1,1,1,1,1,1,1,1,1,1,1,1,1,1,1,1,1,1,1\}$
		\\ \bottomrule[1.5pt]
	\end{tabular}}%

%% file: identification.tex
\section{Automated Identification of Explicit Value-related Requirement Dependencies}
\label{sec_identification}

Automated identification of value-related requirement dependencies and their strengths has not been discussed in the existing literature. Nonetheless, various techniques from information retrieval and data mining domain~\cite{Halpern01062015} can be borrowed to assist such automation. 

This section discusses one of the several possible approaches to automate identification of value-related requirement dependencies. Our proposed approach is based on mining preferences of (potential) users of a software~\cite{do2016incorporating} to identify both the existence and the strengths of explicit value-related dependencies among requirements of a software. 


It has been widely recognize in the literature that users' preferences (customers' preferences) of software requirements can determine their customer values~\cite{do2016incorporating,racheva_business_2010} as highly preferred software requirements are more likely to be purchased and used by the (potential) users. On the other hand, users preferring a requirement $r_j$ may also prefer a requirement $r_i$ (with the probability $p(r_i|r_j)$). This is known as \textit{Market Basket Analysis} or \textit{Association Rule Mining} in data mining domain~\cite{Halpern01062015}. 

An association from a requirement $r_j$ to $r_i$ (users preferring $r_j$ will also prefer $r_i$) can also be interpreted as a causal relation~\cite{sprenger2016foundations} from $r_j$ to $r_i$ meaning that preference (selection) of $r_j$ may cause preference of $r_i$ by the users and therefore give the value of $r_i$. As such, it is clear that a causal relation from $r_j$ to $r_i$ also can be interpreted as a value-related dependency from $r_i$ to $r_j$ (value of $r_i$ depends on preference of $r_j$ by the users).      

Hence, association rule mining of users' preference of requirements can be used for identification of value-related requirement dependencies and the strengths of those dependencies. In this context, measures of causal strength can be used to estimate the strengths of value-related dependencies.      

One of the most commonly adopted measures of causal strength is \textit{Pearl's Measure of Causal Strength}~\cite{sprenger2016foundations,Halpern01062015,pearl2009causality,janzing2013quantifying,eells1991probabilistic} which is denoted by \gls{eta}$_{i,j}$ in $(\ref{Eq_Pearl})$ and derived by $p(r_i|r_j)$. That is the chances that users preferring $r_j$ will also prefer $r_i$. This can be used to estimate the strength of an explicit value-related dependency from $r_i$ to $r_j$. Pearl's measure then can be mapped into a desired fuzzy membership function $\rho(r_i,r_j)$ (which gives the strengths of value-related dependencies in FRIGs) as demonstrated in Figure~\ref{fig_membership}.

Various membership functions could be explored for this mapping based on the preference of the analyst. For instance, the membership function of Figure~\ref{fig_membership_2} treats dependencies with casual strengths below $0.16$ ($\eta_{i,j} < 0.16$) as not worth considering while dependencies with $\eta_{i,j} \geq 0.83$ are treated as full dependencies of strength $1$. Such membership function might be suitable for selection models that formulate dependencies as precedence constraints (BKP-PC models). 

In such models, it might be reasonable to consider a strong causal dependency (say $\eta_{i,j} \geq 0.95$) as a precedence relation rather than ignoring it (BKP-PC models only capture precedence relations). Figure~\ref{fig_membership_3} and Figure~\ref{fig_membership_4} depict other alternative membership functions which unlike membership functions of Figure~\ref{fig_membership_1} and Figure~\ref{fig_membership_2} do not assume linearity for mapping $\eta_{i,j}$ to $\rho(r_i,r_j)$. 

\begin{align}
\label{Eq_Pearl}
&\phantom{ssssss}\eta_{i,j}= p(r_i|r_j)=\frac{p(r_i,r_j)}{p(r_j)},\phantom{s}\eta_{i,j} \in [0,1]
\end{align}
\begin{figure}[!htb]
	\begin{center}
		\subfigure[$$]{%
			\label{fig_membership_1}
			\includegraphics[scale=0.66]{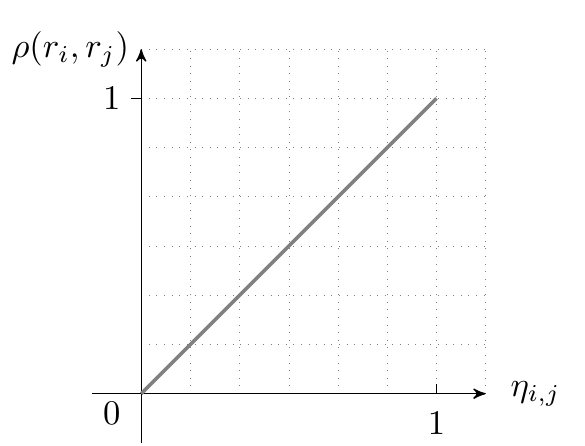}
		}
		\subfigure[$$]{%
			\label{fig_membership_2}
			\includegraphics[scale=0.66]{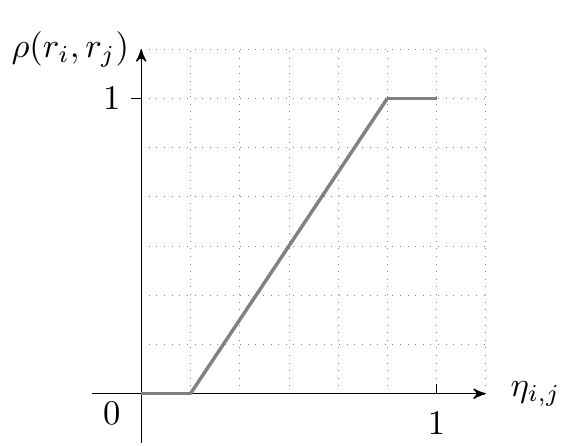}
		}
		\subfigure[$$]{%
			\label{fig_membership_3}
			\includegraphics[scale=0.66]{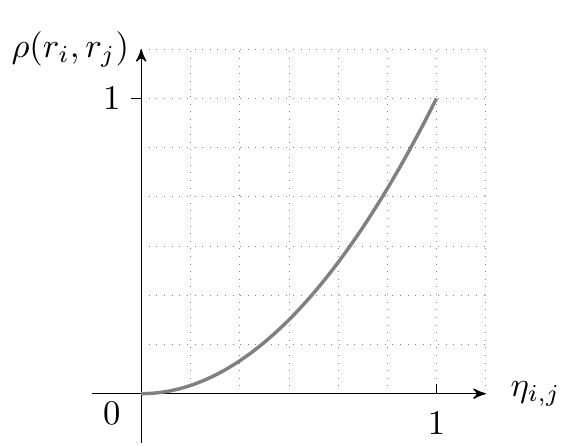}
		}
		\subfigure[$$]{%
			\label{fig_membership_4}
			\includegraphics[scale=0.66]{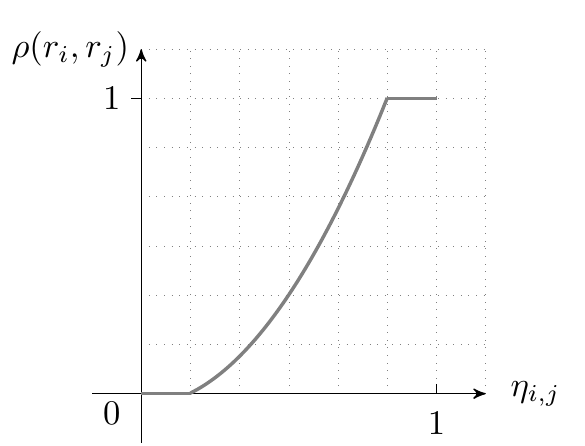}
		}
	\end{center}
	\vspace{-1 em}
	\caption{%
		Sample mappings from $\eta_{i,j}$ to its corresponding membership functions $\rho(r_i,r_j)$
	}%
	\label{fig_membership}
\end{figure}

Finally, users' preferences of software requirements can be gathered in different ways~\cite{leung2011probabilistic,holland2003preference,sayyad2013value} depending on the nature of a software release and the current state of a software. For the first release of a software, users' preferences could be gathered by conventional market research approaches such as conducting surveys or referring to the users' feedbacks or sales records of the similar software products in the market. For the future releases of a software, or when re-engineering of a software is of interest (e.g. for legacy systems) users' feedbacks and sales records of the previous releases of the software might be used in combination with market research approaches to find users' preferences. 

It is also worth mentioning that in cases where collecting users' preferences in large quantities is difficult to achieve, re-sampling methods~\cite{wu1986jackknife} could be used to automatically generate larger samples of users' preferences from a relatively small sample while maintaining the characteristics of the initial sample~\cite{macke2009generating}. 

\begin{figure}[!htb]
	\begin{center}
		\hspace{-1em}\includegraphics[scale=0.67]{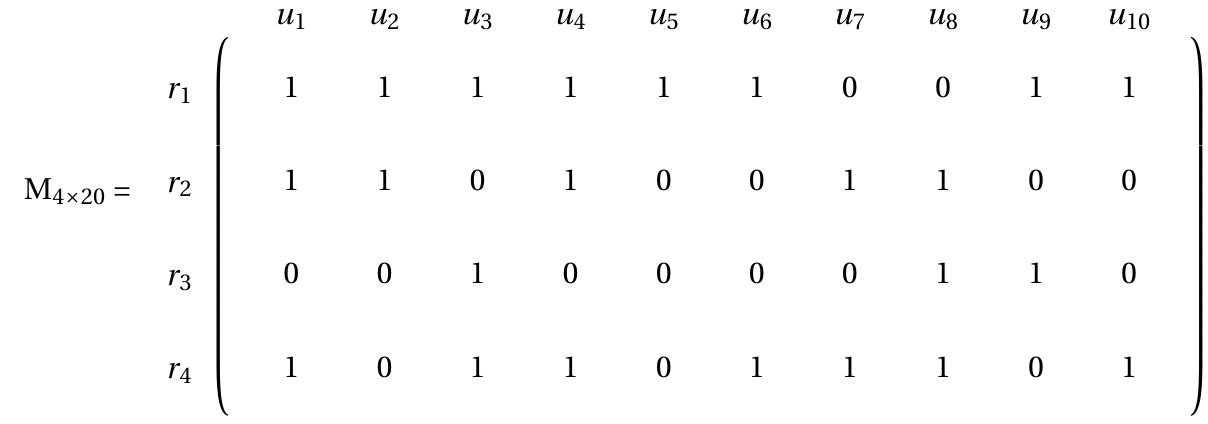}
	\end{center}
	\vspace{-0.5 em}
	\caption{A sample preference matrix $M_{4\times10}$}
	\label{fig_pm}
\end{figure}

\begin{exmp}
	\label{ex_eta}
	Consider the preference matrix $M_{4\times 10}$ of Figure~\ref{fig_pm}. $M_{4\times 10}$ contains $4$ rows and $10$ columns denoting $4$ requirements ($R=\{r_1,r_2,r_3,r_4\}$) and $10$ users ($U=\{u_1,...,u_{10}\}$) respectively. Each element $m_{i,j}$ of $M_{4\times 10}$ specifies whether a user $u_j$ prefers a requirement $r_i$ ($m_{i,j}=1$) or otherwise ($m_{i,j}=0$). For instance, $m_{4,2}=0$ specifies that the requirement $r_4$ is not preferred by the user $u_2$.
	
	Matrix $E_{4\times4}$ (Figure~\ref{fig_eta}) gives Pearl's measure of causal strength computed for pairs of requirements in the preference matrix $M_{4\times 10}$ of Figure~\ref{fig_pm} based on (\ref{Eq_Pearl}). An element $\eta_{i,j}$ of matrix $E_{4\times4}$ denotes the causal strength of an explicit value-related dependence from $r_i$ to $r_j$. For instance, we have $\eta_{1,3}=p(r_1|r_3)=\frac{p(r_1,r_3)}{p(r_3)}=\frac{0.2}{0.3}=0.6667$.
\end{exmp}

\begin{figure}[!htb]
	\begin{center}
		\hspace{-1.2em}\includegraphics[scale=0.45]{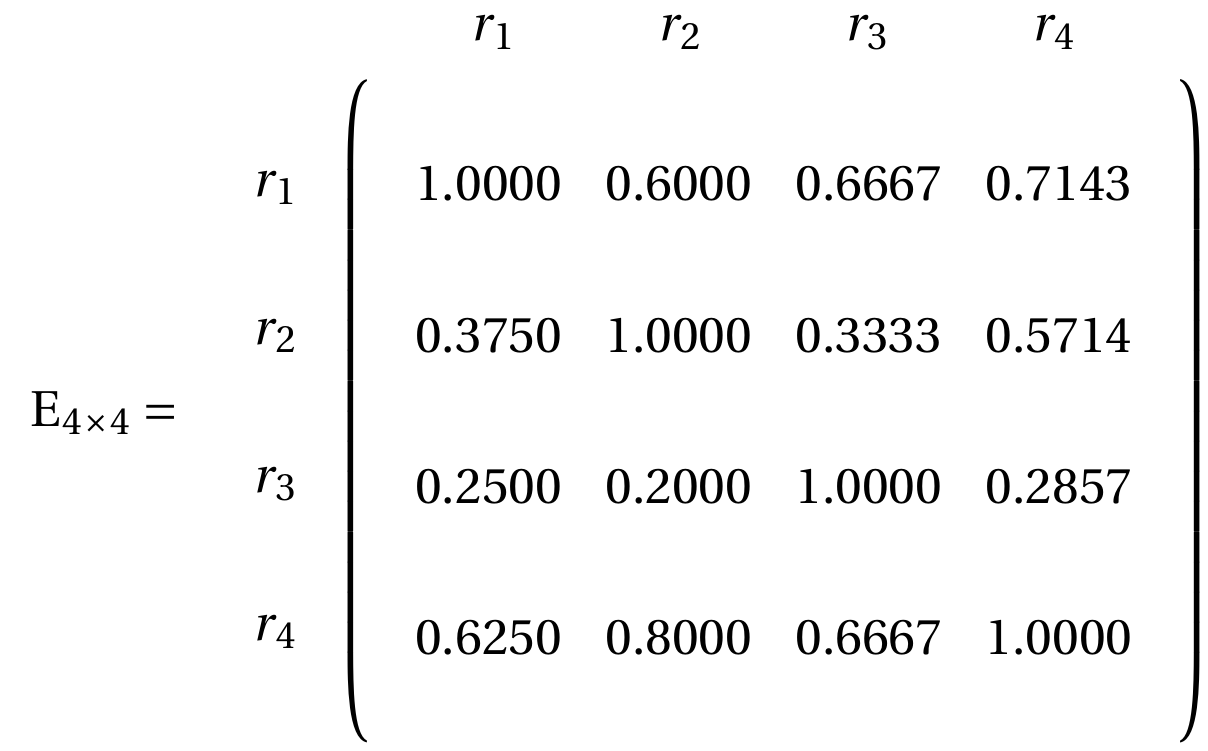}
	\end{center}
	\vspace{-1 em}
	\caption{Pearl's measure of causal strength computed for preference matrix $M_{4\times10}$ of Figure~\ref{fig_pm}}
	\label{fig_eta}
\end{figure}

%% file: conclusion.tex
\section{Conclusions and Future Work}
\label{sec_conclusion}

In this paper we focused on considering the impacts of requirement dependencies on the value of selected requirements (optimal set) during a requirement selection process. To achieve this, we made three main contributions as follows. 

First, we demonstrated using fuzzy graphs for modeling value-related dependencies among software requirements and capturing the strengths of those dependencies. 
Second, we presented a new measure of value referred to as the overall value that factors in the impacts of value-related requirement dependencies on the value of selected requirements (optimal set). 

Finally, we contributed a new requirement selection model referred to as the graph oriented requirement selection (GORS) model that maximizes the overall value of an optimal set by considering the impacts of value-related dependencies on the values of selected requirements. The GORS model not only considers the existence of value-related dependencies but more importantly factors in the strengths of those dependences during a selection process.   

Validity and practicality of our work are verified through a) carrying out several simulations and b) studying a real world software project. The results of our experiments show that: (a) the GORS model can properly capture the strengths of value-related dependencies during a requirement selection while mitigating the selection deficiency problem (SDP), (b) the GORS model always maximizes the overall value of selected requirements, and (c) maximizing the overall and the accumulated values of selected requirements can be conflicting objectives as maximizing one may depreciate the other.


One of the several avenues for extending the present work is to explore various techniques of dependency identification and measures of strength in order to improve the efficiency of automated identification of value-related requirement dependencies and capture various aspects of those dependencies in a software requirement selection process. 

Another possible extension is to also consider cost-related dependencies alongside value-related dependencies among requirements during a selection process. Finally, requirement selection is a NP-hard problem and considering value-related dependencies will add to this complexity. Hence, techniques to enhance scalability of requirement selection models while considering value-related dependencies would be beneficial to the software companies.   
